\documentclass[10pt,twoside,a4paper]{article}
\usepackage[T1]{fontenc}
\usepackage[utf8]{inputenc}
\usepackage[a4paper,margin=3cm,includehead]{geometry}%for notes
\usepackage{amssymb}
\usepackage{epstopdf}

\usepackage[overload]{empheq}
\usepackage{latexsym,dsfont,textcomp,amsmath}
\usepackage[english]{babel}
\usepackage{enumerate,indentfirst}
\usepackage[colorlinks,citecolor=blue,urlcolor=blue]{hyperref}
\usepackage{amssymb}
\usepackage{amsthm}
\usepackage{hyperref}

\def\R{\mathbb R}
\def\N{\mathbb N}
\def\Z{\mathbb Z}

\def\E{\mathbb E}

\def\shb{{\cal B}}

\newcommand{\indep}{\mathrel{\text{\scalebox{1.07}{$\perp\mkern-10mu\perp$}}}}

\newtheorem{thm}{Theorem}[section]

\newtheorem{lemma}{Lemma}[subsection]

\newtheorem{defi}[thm]{Definition}

\usepackage{etoolbox}
\makeatletter
\patchcmd{\ps@pprintTitle}{\footnotesize\itshape
       Preprint submitted to \ifx\@journal\@empty Elsevier
       \else\@journal\fi\hfill\today}{\relax}{}{}
\makeatother

\date{}
\title{Conditional heteroskedasticity in crypto-asset returns.}
\author{Charles Shaw \\ \\
\small{University of London, Birkbeck College (student)}\\
\small{Accepted for publication in} \\
\small{Journal of Statistics: Advances in Theory and Applications, 20(1).}\\
\small{email: cshaw11@mail.bbk.ac.uk}
}

\begin{document}
\maketitle
\begin{abstract}
In a recent contribution to the financial econometrics literature, Chu et al. (2017) provide the first examination of the time-series price behaviour of the most popular cryptocurrencies. However, insufficient attention was paid to correctly diagnosing the distribution of GARCH innovations. When these data issues are controlled for, their results lack robustness and may lead to either underestimation or overestimation of future risks. The main aim of this paper therefore is to provide an improved econometric specification. Particular attention is paid to correctly diagnosing the distribution of GARCH innovations by means of Kolmogorov type non-parametric tests and Khmaladze's martingale transformation. Numerical computation is carried out by implementing a Gauss-Kronrod quadrature. Parameters of GARCH models are estimated using maximum likelihood. For calculating P-values, the parametric bootstrap method is used. Further reference is made to the merits and demerits of statistical techniques presented in the related and recently published literature.

\end{abstract}

\newpage

\section{Introduction}

A cryptocurrency, such as Bitcoin, is a digital decentralized currency that makes use of cryptography to regulate the creation and transactions of the exchange unit. It is an emerging, retail-focused, highly speculative market that lacks a legal and regulatory framework comparable to other asset classes. Cryptocurrencies are decentralized in the sense that it is not created by any central authority and may, in principle, be immune to any central bank's interferences. At the time of writing this paper, it is estimated that the transaction volume in cryptocurrencies exceeds 100 million USD per day. The number of hedge funds that trade cryptocurrencies has recently reached approximately 100 for the first time (\cite{R18}), of which more than three-quarters were launched in 2017. The increase, from 55 hedge funds at Aug. 29 to 110 hedge funds at Oct. 18, comes as investors pile into the high-octane cryptocurrency market, which has seen a tenfold increase in its value in 2017. 

Although the cryptocurrency market is still relatively new and undeveloped, there have been a number of interesting developments. Just by way of illustration, in Q4 of 2017 alone, the following occurred: JP Morgan confirmed heavy investment in blockchain technology, which underpins cryptocurrency transactions; CME Group, the world's largest futures exchange operator, announced the launch of trading in Bitcoin derivatives at the end of 2017, pending regulatory review; Swiss bank Vontobel, the country's second-biggest provider of structured products comes after CME Group, announced the launch of Bitcoin futures on the Swiss stock exchange. Such involvement on behalf of institutional market participants makes it interesting to study this newly emerging asset class. 

From a regulatory perspective, the Financial Conduct Authority (FCA) has expressed concern that retail investors have increasingly been buying Bitcoin contracts for difference (CFDs). The FCA listed price volatility, leverage, charges and funding costs, and price transparency as four risks to investing in crypto-based CFDs. This paper examines the first of these risks: volatility.\footnote{ The decision to focus on price volatility is largely motivated by data availability, since it has been impossible for me to obtain data on other risks. However, we now have fairly trustworthy closing price data from four main cryptocurrency exchanges. And this data is what is used in this study.}

\subsection{Related literature}
Understanding price volatility dynamics is of considerable interest to those seeking to understand the price dynamics of a financial assets. To this end, there is a well-developed body of research on econometric inference techniques for (mostly second order) stationary financial data. However, there is a paucity of research on cryptocurrency volatility modelling. Of particular interest is the recent work on volatility of cryptocurrencies by Chu, Chan, Nadarajah, and Osterrieder (2017, \cite{Chu17}), which provides the first modelling of the seven most popular cryptocurrencies. The aim of this paper is to extend their work and propose an alternative, and arguably more robust, econometric specification.

In related literature, \cite{Gro14} empirically analyses Bitcoin prices using an autoregressive jump-intensity GARCH model and finds strong evidence of time-varying jump behaviour. \cite{Chap} test for the optimal number of states for a Markov regime-switching (MRS) model to capture the regime heteroskedasticity of Bitcoin. \cite{K17} run a model comparison exercise according to three information criteria, namely Akaike (AIC), Bayesian (BIC) and Hannan-Quinn (HQ) and find that the AR-CGARCH model gives the best fit for Bitcoin. \cite{CZ18} fit more than 1,000 GARCH models to the log returns of the exchange rates to find that two-regime GARCH models produce better VaR and ES predictions than single-regime models for four of the main cryptocurrencies, namely Bitcoin, Ethereum, Ripple and Litecoin. \cite{STT17} investigate Bitcoin for the BTC/USD exchange rate using high-frequency (transaction-level) data obtained from Mt. Gox exchange, the leading platform during the sample period of June 2011 to November 2013, and note the asset's extreme volatility and apparent discontinuities in the price process. They assert two empirical observations. First, they argue that jumps are an essential component of the price dynamics of the BTC/USD exchange rate: out of the 888 sample days, they identify 124 jump days. Second, they show that jumps cluster in time: they find that runs of jump days that are incompatible with the assumption of independent Poisson arrival times. They conclude that order flow imbalance, illiquidity, and the dominant effect of aggressive traders are significant factors driving the occurrence of jumps. 

At first, these findings seems intuitive. Cryptoassets, by virtue of their design, do not rely on the stabilizing policy of a central bank. As a result, the reaction to new information - whether this information is spurious or fundamental - are prone to demonstrate high volatility relative to established assets. This volatility is amplified by the relative illiquidity of the market. In addition, the absence of official market makers would make cryptoassets fragile to large market movements.

Using a GARCH (1,1) model, \cite{C17} examined Bitcoin's volatility in respect to the macroeconomic variables of countries where it was being traded the most. It was argued that if the volatility levels follow the same trend as in the last six and a half years, Bitcoin may match the fiat currency levels of volatility in 2019-2020. Building on this work, \cite{Chu17} fitted 12 GARCH-type models to seven major cryptocurrencies. The the distribution of the innovation process were taken to be one of normal, skew normal, Student's $t$, skew Student's $t$, skew generalized error distribution, inverse Gaussian, and generalized hyperbolic distribution. Model selection criteria were then used to pick the best fit. They found that Gaussian innovations provide the smallest values of AIC, AICc, BIC, HQC and CAIC for each cryptocurrency and each GARCH-type model. Further, \cite{Chu17} make use of the skewed generalized error distribution (SGED). 

This paper will demonstrate that using a skewed generalized error distribution is a poor modelling choice. This is because the moment generating function of a SGED does not exist under some important conditions. Using Student's $t$ to model the innovation process also represents a poor choice for financial engineering applications since the distribution does not possess a moment generating function. If innovations followed a Student's $t$ distribution under a risk-neutral measure then the value of a call option would be infinite.

\section{Preliminaries}
Given a price process $S_i$, we define $(X_i)_{i \in \N}=\ln (S_i/S_{i-1})$ to be the daily log-returns of observed data series, indexed by time index $i$, where
\begin{equation}
\label{eq:ARIntro}
X_i=\mu_i+\sigma_i \epsilon_i, ~ i>r, 
\end{equation} 
is driven by some innovation process $\epsilon_i$, for $r \geq 1$. This allows $\mu_i$ and $\sigma_i$ to depend on ${\cal F}_{i-1}=\sigma \{H_i,X_1,\ldots,X_{i-1}\}$, where ${\cal F}_{i-1}$ is the $\sigma-$algebra induced by variables that are observed at time $i-1$, and $H_i$ is a random vector with plausibly exogenous variables. Then, $(\epsilon_i)_{i>r}$ is a normal random i.i.d. sequence satisfying the standard assumptions $\E[\epsilon_i]=0$ and $Var(\epsilon_i)=1$. Note that $\epsilon_i \indep {\cal F}_{i-1}, ~ \forall ~ i>r.$ We further assume that $\epsilon_1, \ldots, \epsilon_n$ are observations from the process $(\epsilon_i, i \in \Z)$\footnote{ Without loss of generality, the index $i$ can be assumed to take values in either $\N_0$ or in $\Z$.} to be strictly  stationary, ergodic, and nonanticipative. The requirement of causality is often added to the set of assumptions. However, we are able to get this for "free" since every strictly stationary GARCH process is causal.

The above framework was used by \cite{Engle82} to introduce the so-called Auto-Regressive Conditional Heteroskedasticity (ARCH) class of processes which stemmed from \eqref{eq:ARIntro} with non-constant $\sigma_i$. These allowed the conditional variance $\sigma^2_i$ to depend on lagged values of $(X_i-\mu_i)^2$. \cite{Boll86} generalized this framework with the so-called GARCH models, and since then there has been a range of models that fall under the broad umbrella of (Generalized) Autoregressive conditional heteroskedasticity i.e. GARCH literature. An important feature in GARCH models is that $\sigma_i$ depends on its own past values. Further, the conditional distribution $\{X_i |{\cal F}_{i-1}\}$ is such that 
$\{X_i |{\cal F}_{i-1} \} \sim N (\mu_i, \sigma^2_i).$ In this case we say that the innovations are Gaussian and that $ \E\{(X_i-\mu_i)^2 |{\cal F}_{i-1} \} = \sigma^2_i.$ Since $\sigma$ is not observable per se, we need $\sigma_n$ to admit stationarity. In other words, we need $\E [\sigma^2_n]$ to converge to some positive constant as $n \to \infty$. Next, we outline the main models in this family that are relevant for the purposes of this study.

\subsection{GARCH(1,1)}
Setting $r=1$, $\mu_i=\mu$, and $H=\sigma_1$, yields GARCH (1,1)\footnote{ It is evident that an ARCH(1) model can be derived from \eqref{eq:GAR} by simply setting $\beta=0$.}:
\begin{equation}
\label{eq:GAR}
\sigma^2_i=\omega+\alpha (X_{i-1}-\mu)^2+\beta \sigma^2_{i-1}=\omega + \sigma^2_{i-1} (\alpha \epsilon^2_{i-1}+\beta), \qquad \forall \quad i\geq 2,\quad \alpha \geq 0,\quad \beta \geq 0, \quad \omega >0. 
\end{equation}
By taking expectations of $\sigma^2_n$, it is possible to show its stability conditions:
\begin{eqnarray*}
\label{eq:GAR1}
\E [\sigma^2_n]&=& \omega +\alpha \E \Big[\E\{X_{n-2}-\mu)^2 |{\cal F}_{i-1} \} +\beta \E (\sigma^2_{i-1})  \Big] \\
&=& \omega (\alpha+\beta) \E (\sigma^2_{i-1}) \\
&=& \omega \Big[ \frac{1-(1-k)^{n-1}}{k}+(1-k)^{n-1}(\sigma^2_{1})   \Big]
\end{eqnarray*}
where $k=1-\alpha-\beta$. If $k>0$ then we can see that $\E [\sigma^2_n]$ stabilizes as $n \to \infty$: $
\lim_{n \to \infty} \E[\sigma^2_n] = \frac{\omega}{k}$

\begin{lemma} 
\label{Remark1} Let us consider two processes $\sigma$ and $\nu$, with starting conditions $\sigma_1$ and $\nu_1$, both driven by a single innovation process $\epsilon$:
\begin{eqnarray*}
\label{eq:GAR2}
|\sigma^2_{n+1}-\nu^2_{n+1}| &=& |\sigma^2_{n}-\nu^2_{n}| (\alpha \epsilon^2_n+\beta)  \\
\label{eq:GAR21}&=& |\sigma^2_{n}-\nu^2_{n}| \prod_{k=1}^n (\alpha \epsilon^2_k+\beta).
\end{eqnarray*}
The condition $k>0$, needed for finiteness of first moment, is stronger than the stationarity condition $\mathbb{E}\Big[\ln(\alpha \epsilon^2_k+\beta)\Big]<0$.
\end{lemma}

\begin{proof}
If  $\E \Big[\ln(\alpha \epsilon^2_k+\beta)\Big]<0$ then, by application of L.L.N. to sequence $\ln(\alpha \epsilon^2_k+\beta), \quad k \geq 1$, we can see that
\[
|\sigma^2_{n}-\nu^2_{n}| \prod_{k=1}^n (\alpha \epsilon^2_k+\beta) \to 0.
\]
Further, by application of S.L.L.N.,
\[
\frac{1}{n}\sum^n_{k=1}  \Big[\ln(\alpha \epsilon^2_k+\beta) \Big] \to \E  \Big[\ln(\alpha \epsilon^2_k+\beta)\Big]<0.
\]
Therefore, for a given $\alpha >0$,
\[
\frac{1}{n}\sum^n_{k=1} \Big[ \ln(\alpha \epsilon^2_k+\beta) \Big] \to \E  \Big[\ln(\alpha \epsilon^2_k+\beta) \Big]< - na
\]
for almost all $n \geq 1$, and 
\[
\prod_{k=1}^n \Big[\alpha \epsilon^2_k+\beta\Big] < e^{-na}
\]
for almost all $n \geq 1$. Since the log function is concave, we apply Jensen's inequality to yield: 
\begin{eqnarray*}
\E \Big[ \ln(\alpha \epsilon^2_k+\beta) \Big] \leq \ln  \E  \Big[\alpha \epsilon^2_k+\beta \Big].
\end{eqnarray*}
Hence the condition $k>0$ is stronger than the stationarity condition $\E \Big[\ln(\alpha \epsilon^2_k+\beta) \Big]<0$, as required.
\end{proof}

\subsection{GARCH(p,q)}
It is possible to introduce more lags for $X$ and $\sigma$. Consider $r=max(p,q)$, $H=(\sigma_1,\ldots,\sigma_r)$, $\mu_i=\mu$, and
\begin{equation*}
\sigma^2_i= \omega +\sum^p_{k=1} \beta_k \sigma^2_{i-k}+\sum^q_{j=1} \alpha_j (X_{i-j-\mu)^2}, \quad i>r
\end{equation*} 
where $\alpha_j, \beta_j \geq 0$, $j \in \{1, \ldots, r \}$, and $\omega >0$. Therefore, $
\E [\sigma^2]=\omega+ \sum^r_{j=1} \lambda_j \E [\sigma^2_{i-j}],$ where $\alpha_j=0 \quad \forall j>q$, $\beta_j=0 \quad \forall j>p$, and $\lambda_j=\alpha_j+\beta_j$. Asymptotic properties are discussed in next Lemma.

\begin{lemma}
Let 
$ k=1-\sum_{j=1}^{\max(p,q)} \lambda_j.$
Then, 
\begin{equation*}
\lim_{n \to \infty} \E[\sigma^2_n]=\left\{
                \begin{array}{ll}
                  \omega/k, & \qquad \forall k > 0\\
                  +\infty   & \qquad \forall k \leq 0
                \end{array}
              \right.
\end{equation*}
\end{lemma}

\begin{proof}

Proof follows \cite{Remi13}, which expands on the proof partly shown in \cite{Boll86}. First, we define $K(z)= z^r \sum_{j=1}^r \lambda_j z^{r-j}.$ The aim is to demonstrate that $k= K(1) > 0 \implies \lim_{n \to \infty} \E[\sigma^2_n] = \frac{\omega}{K(1)}$ and $ k=K(1) \leq 0 \implies \lim_{n \to \infty} \E[\sigma^2_n] =+ \infty.$ Let $U \equiv \{z\subset \mathbb{C}:|z|<1\}$ be a unit ball on a complex plane. $K(1)>0$ implies that all roots of the polynomial $K(z)$ are within $U$. These roots are the eigenvalues of the matrix 
\begin{equation*}
A=
\begin{bmatrix}
    \lambda_{1}   & \lambda_{2}  &  \dots & \lambda_r \\
    1             & 0            & \dots  &  0 \\
    0             & 1            & \dots  &  0 \\
    \vdots        & \vdots       & \vdots &  0 \\
    0 & \dots     &   1 &0
    \end{bmatrix}.
\end{equation*}
For exposition, we require the spectral radius of $A$. 

\begin{defi}
Let $X$ be a Banach space and let $f \colon \shb \to \mathbb{R}$ be a bounded linear functional on $X$. The \textit{norm} of $f$, denoted by $\|f\|$, is defined by $ \|f\| = \inf \{ L \in [0,+\infty) :          |f(x)| \leq L \|x\| \mbox{ for all } x \in X \}.$ If $\|\cdot\|$ is any norm of the set of $r \times  r$ matrices, then the \textit{spectral radius} of $A$, denoted by $\rho(A)$, is defined by $\rho(A)=\lim_{n \to \infty} \|A^n\|^{1/n}.$ 
\end{defi}

Now, the statement $K(1)>0$ is equivalent to saying that the spectral radius of $A$ is smaller than 1. If the roots of $K$ are inside the unit ball, then $K(1)>0$. If $K(1)=0$ then 1 is a unit root not in the unit ball, which is a contradiction. If $K(1)<0$ then $K(z) \to \infty$ as $z \to \infty$ so there would be a real root of $K$ greater than 1. Lets say that now $K(1)>0$. Then $\exists z_0 \in (0,1)$ so that $K(z)>1, \quad \forall  z \in [z_0,1]$. We fix $z$, and let
\begin{equation*}
v_i=
\begin{bmatrix}
\E[\sigma^2_{i}] \\
z\E[\sigma^2_{i-1}] \\
z^{r-1} \E[\sigma^2_{i-r+1}]
\end{bmatrix} , \quad \makebox{and} \quad \tilde{\omega}
\begin{bmatrix}
\omega \\
0 \\
\vdots \\
0
\end{bmatrix}.
\end{equation*}
$ \E[\sigma^2_{i}]= \omega +\sum^r_{j=1} \lambda _j \E[\sigma^2_{i-j}], \quad \forall \quad i \geq r+1$ is equivalent to $v_i=\tilde{\omega} A_z v_{i-1} \quad \forall \quad i \geq r+1$, where
\begin{equation*}
A_z=
\begin{bmatrix}
\lambda_1 & \lambda_2/z & \cdots & \lambda_r/z^{r-1} \\
z & 0 & \cdots & 0 \\
0 & z & \cdots &  0 \\
\vdots & \ddots & \ddots & 0 \\
0 & \ldots & z & 0
\end{bmatrix}.
\end{equation*}
Define the norm $\|B\| = \max_{1\leq i \leq r} \Big( \sum^r_{j=1} |B_{ij}| \Big)$. We get $\|A_z=z<1 \|$, since $K(z)>0$. It follows that $ v_n=(I+A_z+ \ldots +A_z^{n-r-1}) ~ \tilde{\omega}+A_z^{n-r} v_r.$ Next, define $|x|= \max_{1\leq i \leq r} |x_j|.$ Then, for any matrix $B$, $|B_x| \leq \|B\| |x|$. Therefore $|A^{n-r}_z v_r| \leq \|A^{n-r}_z\| |v_r| \leq z^{n-r} |v_r| \to 0$ as $n \to \infty$. Since $\|A\| <1$, $I+A_z+ \ldots +A_z^{n-r-1} \to (I-A_z)^{-1}$, as $n \to \infty$. Thus, $v_n$ converges to $(I-A_z)^{-1} \tilde{w}=wy$, where $(I-A_z)y=e_1=(1,0,\ldots,0)'$. Therefore, $y_j=z^{j-1}y_1$, and $1=y_1-\sum_{j=1}^r \lambda_j z^{-j+1}y_1=y_1K(1)$. Therefore, $y_1=1/K(1)$ and $\E[\sigma^2_n]$ converges to $\omega/K(1)$. Next, we demonstrate that 
\begin{equation}
\label{eq:K0}
k= \leq 0 \implies \lim_{n \to \infty} \E[\sigma^2_n] = + \infty
\end{equation}
and
\begin{equation}
\label{eq:K1}
k=K(1)>0 \implies \lim_{n \to \infty} \E[\sigma^2_n] = \frac{\omega}{K(1)}.
\end{equation}
We examine \eqref{eq:K0} first to show that $k= \leq 0$ implies that $\lim_{n \to \infty} \E[\sigma^2_n] =  \infty$.\\

Case 1. Let $K(1)=0$. Then $A$ is the transition matrix of an irreducible Markov chain s.t. $A^n \to B$, with 
\begin{equation*}
B_{ij}=\pi_j=\frac{\sum^r_{k=j} \lambda_k}{\sum^r_{k=1} k\lambda_k}, \quad \forall i,j \in \{1, \ldots, r \}.
\end{equation*}
Therefore, if $\sum^r_{k=j} \lambda_k=1$, then $ \E[\sigma^2_n] \to \infty$ as $n \to \infty$.\\

Case 2. Let $K(1)<0$, and $a=\sum^r_{j=1} \lambda_j>1.$
Therefore, 
\begin{equation*}
\E[\sigma^2_n]> \frac{\E[\sigma^2_n]}{a}=\frac{\omega}{a}+\sum^r_{j=1} \frac{\lambda_j}{a} \E[\sigma^2_{n-j}].
\end{equation*}
Let $u_n$ be the solution of $u_n = \frac{\omega}{a}+\sum^r_{j=1} u_{n-j}, \quad n>r, \quad j \in \{ 1,\ldots,r \}$ and $u_j=\E[\sigma^2_n]$. Then $\E[\sigma^2_n]-u_n>0$ $\forall$ $n>r$. Since $\sum^r_{k=j} \lambda_k=1$, $u_n/n$ converges to a positive number. Hence $\E[\sigma^2_n] \to \infty$ as $\to \infty$, as required. 
\end{proof}

\subsection{Exponential GARCH (EGARCH)}

The EGARCH models, introduced by \cite{Nel91}, allow $r=1$, $H=\sigma^2_1$, and $ \ln(\sigma^2_i)=\omega+\alpha [|\epsilon_{i-1}|-\E(|\epsilon_{i-1}|)]+\gamma \epsilon_{i-1}+\beta \ln (\sigma^2_{i-1}), \quad i\geq 2.$ In this framework, $\E[\ln(\sigma^2_i)]$ convergence to a (finite) limit exists if and only if $|\beta|<1$. This limit is $\omega/(1-\beta).$

\subsection{Nonlinear GARCH (NGARCH)}
A recognised problem with the standard GARCH models is their inability to differentiate negative and positive innovations. We tend to observe that volatility changes are more pronounced after a large negative shock when compared to an equally large positive shock. This is the so-called leverage effect. The  NGARCH models, introduced by \cite{ENG93}, aim to model such asymmetry of volatility behaviour in previous specifications. The models is now set up as $
\sigma^2=\omega+\alpha (\epsilon_{i-1}-\rho)^2 \sigma^2_{i-1} +\beta \sigma^2_{i-1}, \quad i\geq 2.$ The restrictions for the positivity of $\sigma^2_i$ are $\omega>2$, $\alpha,\beta \geq 0$. Parameter $\rho$ is the leverage effect. The limit of $\E[\sigma^2_n]$ is $\omega/k$ for $k>0$. $\E[\sigma^2_n]$ converges if and only if $1-k=\alpha(1+\theta^2)+\beta<1$. An alternative version of the NGARCH model was originally estimated by Engle and Bollerslev \cite{EB86}, $\sigma^2_i=\omega+\alpha |\epsilon_{i-1}|^\delta+\beta \sigma^2_{t-1}.$ With most financial assets, the rates of returns are estimated $\delta<2$, although not always significantly so.

\subsection{Glosten-Jagannathan-Runkle (GJR)-GARCH}

The motivation behind GJR-GARCH models, introduced by \cite{Glo93}, was to model the asymmetric behaviour of volatility when innovations are negative or positive. In these models, $r=1$, $H=\sigma_1$:
\[\sigma^2_i = \omega +\alpha \sigma^2_{i-1} \epsilon^2_{i-1}+\beta \sigma^2_{i-1} +\gamma \sigma^2_{i-1} \{\sup(0,-\epsilon^2_{i-1})\}^2, \quad i \geq 2.\]
Again, the limit of $\E[\sigma^2_n]$ is $\omega/k$ for $k>0$. $\E[\sigma^2_n]$ converges if and only if $k=1-\alpha-\beta-\gamma/2>0$.

\subsection{Augmented GARCH}

The Augmented GARCH models, introduced by \cite{D97}, contain all the GARCH specifications mentioned previously. The assumption $\epsilon_i  \sim N(0,1)$ is relaxed, and it is assumed that only the common distribution of $\epsilon_i$ are mean 0 and unit variance. Under the assumption $\alpha_i \geq 0 $ for $i \in \{0,\ldots,5\}$, the model is described as follows. Let us define
\begin{equation*}
f_\delta (x)=\left\{
                \begin{array}{ll}
                  \frac{x^\delta-1}{\delta} & \qquad \forall \delta>0\\
                  \ln (x)   & \qquad \forall \delta=0
                \end{array}
              \right.
\end{equation*}
A similar consideration yields
\begin{equation*}
f^{-1}_\delta (x)=\left\{
                \begin{array}{ll}
                 (x\delta+1)^{1/\delta} & \qquad \forall \delta>0\\
                  e^x   & \qquad \forall \delta=0
                \end{array}
              \right.
\end{equation*}
Then, define the volatility process $\sigma$ as $\sigma^2_i=f^{-1}_\lambda (\phi_i-1)$ $i \geq 2$,  where
\begin{eqnarray}
\xi_{1,i} & = & \alpha_1+\alpha_2 |\epsilon_i-c|^\delta +\alpha_3\{\sup(0,c-\epsilon_i)\}^\delta,  \\
\xi_{2,i} & = & \alpha_4 f_\delta (|\epsilon_i-c|)+\alpha_5 f_\delta \{\sup(0,c-\epsilon_i\}, \\
\phi_i & = & \alpha_0 +\phi_{i-1} \xi_{1,i-1}+\xi_{2,i-1}.
\end{eqnarray}
Strict stationarity is achieved if $\E[\ln (\xi_{1,i-1})]<0$ and $\E[\max(\xi_{2,i-1})]<\infty.$ Further, by Jensen's inequality, $\E (\ln (\xi_{1,i-1})<0$ if $\E[\xi_{1,i-1}]<1$. Further discussion of the time series particulars of models a-la GARCH in the context of cryptocurrency modelling is presented in \cite{Chu17}.\footnote{ For a classic treatment of GARCH type models, the interested reader is referred to \cite{Bro91}, \cite{GM95}, and \cite{GM96}. A more up to date exposition is available \cite{FZ10}.}

\newpage
\section{Data}
Recent daily market capitalization figures for all cryptocurrencies can be accessed via \cite{Coin17}, who provide daily cryptocurrency data (transaction count, on-chain transaction volume, value of created coins, price, market cap, and exchange volume) in CSV format. The daily data sample of the coins selected is from 22-Jun-14 to 24-Mar-18. Other notable cryptocurrencies such as Ethereum etc were omitted due to lack of sufficient time-series data available. Source code for the data collection tools is available at Github (\url{https://github.com/whateverpal/coinmetrics-tools}). The data is transformed into its log returns. Although there is a deluge of new cryptocurrencies (around 1500 cryptocurrencies recorded at the time of writing), vast majority are new and have extremely short time-series observations. However, there is sufficiently long time-series available for the cryptocurrencies in our sample. When taken together, this sample represented around 66 per cent of the total market at the end of 2017. This figure at end of March stood at 56 per cent. At the end of 2017 Bitcoin alone dominated 62.5 per cent of the market. This figure at end of March stood at 44 per cent, see Table \ref{sample}.

\begin{table}[]
\centering
\caption{Cryptocurrency Market Capitalizations in USD: Jun 22, 2014 - Mar 24, 2018. Data accessed 24-Mar-18.}
\label{sample}
\begin{tabular}{llllll}
Name         & Symbol & Market Cap        & Price      & \% 7d   & Share of Market \\ \hline
Bitcoin      & BTC    & \$146,064,176,501 & \$8,622.92 & 11.93\% & 43.98\%         \\
Ripple       & XRP    & \$25,145,211,528  & \$0.64     & 7.87\%  & 7.57\%          \\
Litecoin     & LTC    & \$9,005,259,992   & \$161.43   & 11.35\% & 2.71\%          \\
Dash         & DASH   & \$3,343,681,609   & \$419.69   & 15.18\% & 1.01\%          \\
Monero       & XMR    & \$3,339,885,427   & \$210.54   & 12.31\% & 1.01\%          \\
Dogecoin     & DOGE   & \$400,962,064     & \$0.00     & 18.35\% & 0.12\%          \\
MaidSafeCoin & MAID   & \$132,063,392     & \$0.29     & 13.73\% & 0.04\%         
\end{tabular}
\end{table}

\begin{figure}[htbp]
 \subsection{Daily log returns, ACF \& PCF, kurtosis in logarithmic returns - Bitcoin,	Dashcoin, Dogecoin,	and Litecoin}
 \caption{Data for Bitcoin,	Dashcoin, Dogecoin,	and Litecoin from June 22, 2014 - March 24, 2018. (Left) Graphs of the daily log returns of the exchange rates. (Middle) Sample autocorrelation function (ACF) and partial autocorrelation function (PACF) for the return series. (Right) The solid blue line is the empirical Pdf of log returns. The dotted red line is the normal Pdf. Initial signs indicate presence of volatility clustering and kurtosis in logarithmic returns.} 
\centering
\includegraphics[scale=.3]{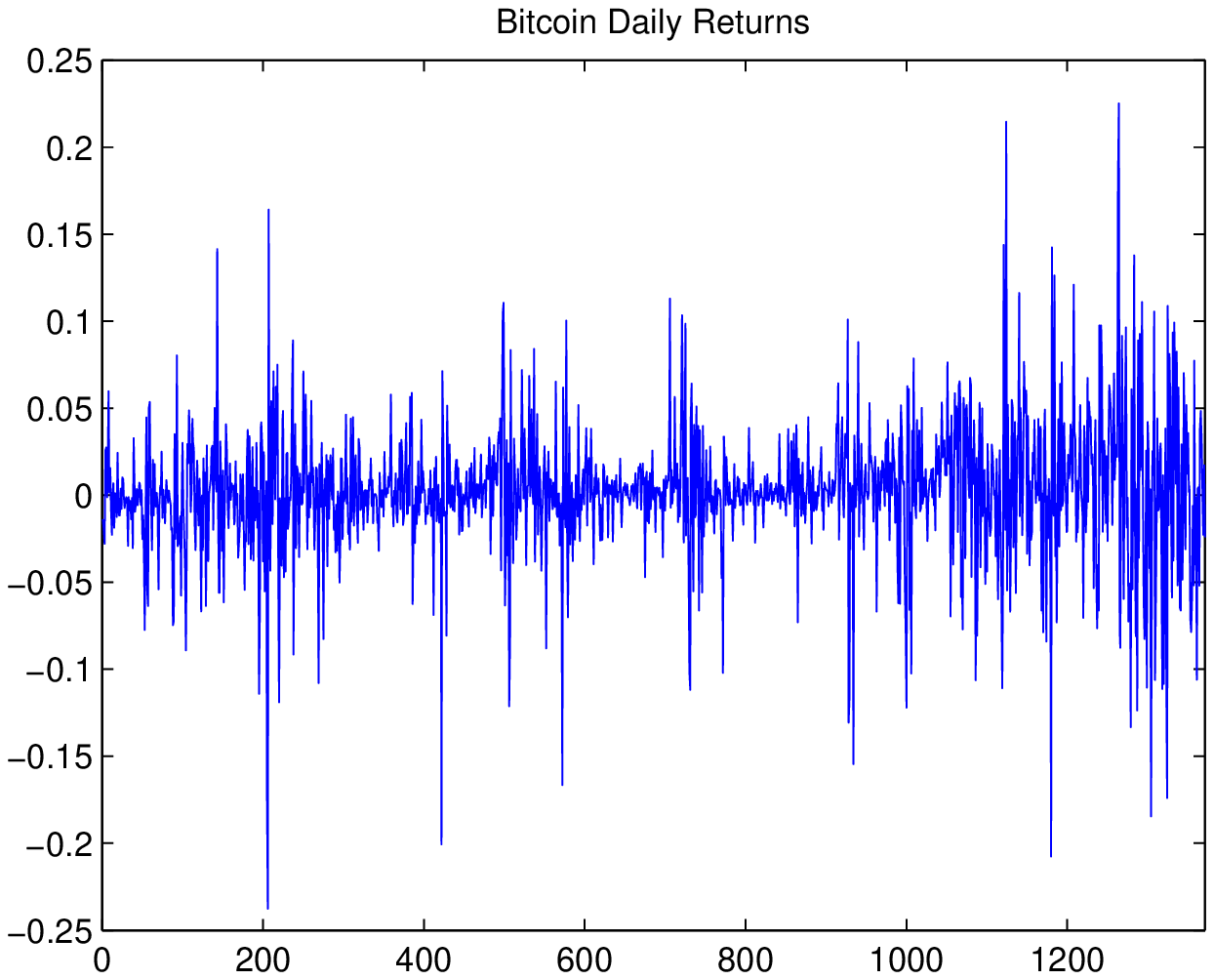}
\includegraphics[scale=.3]{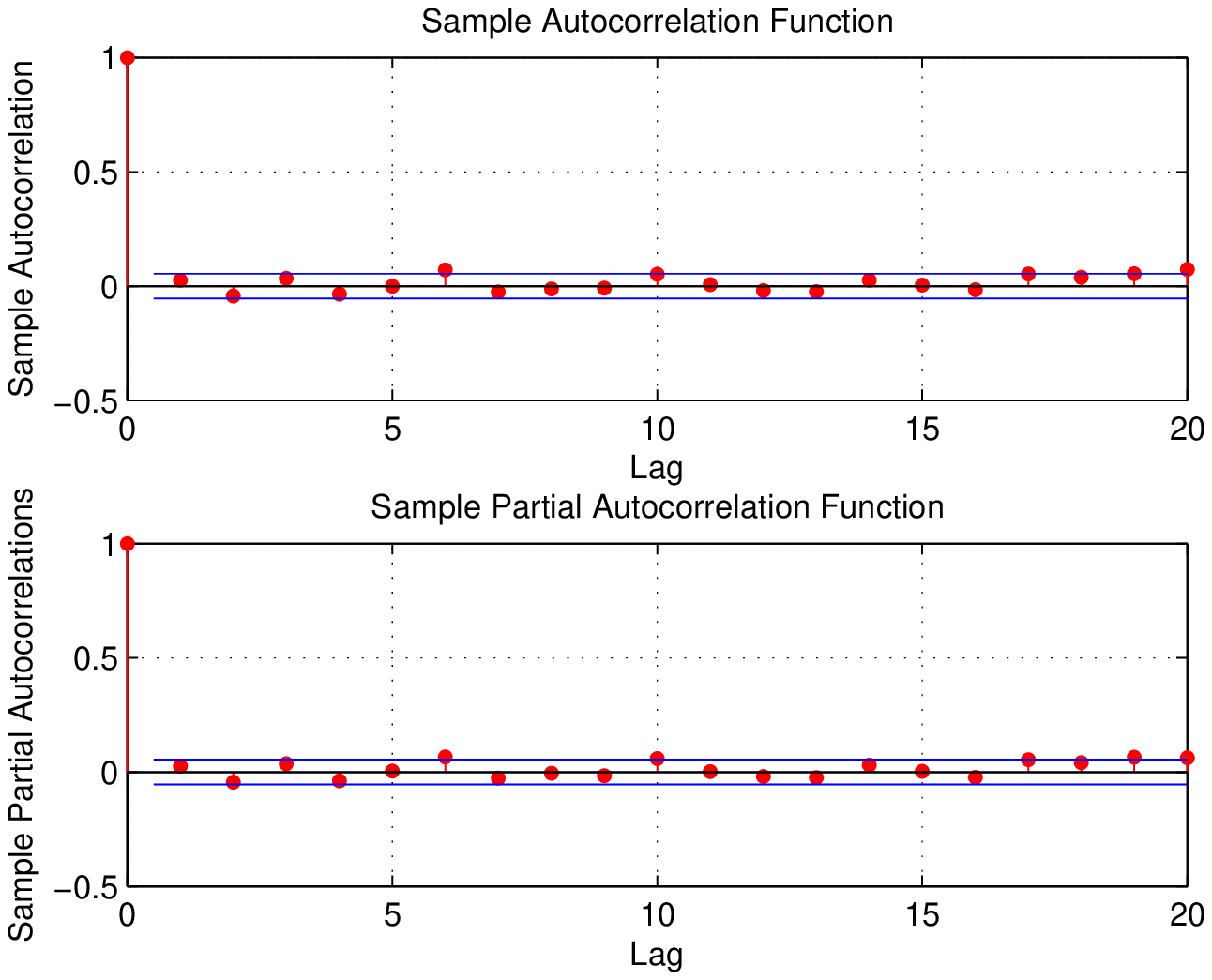}
\includegraphics[scale=.3]{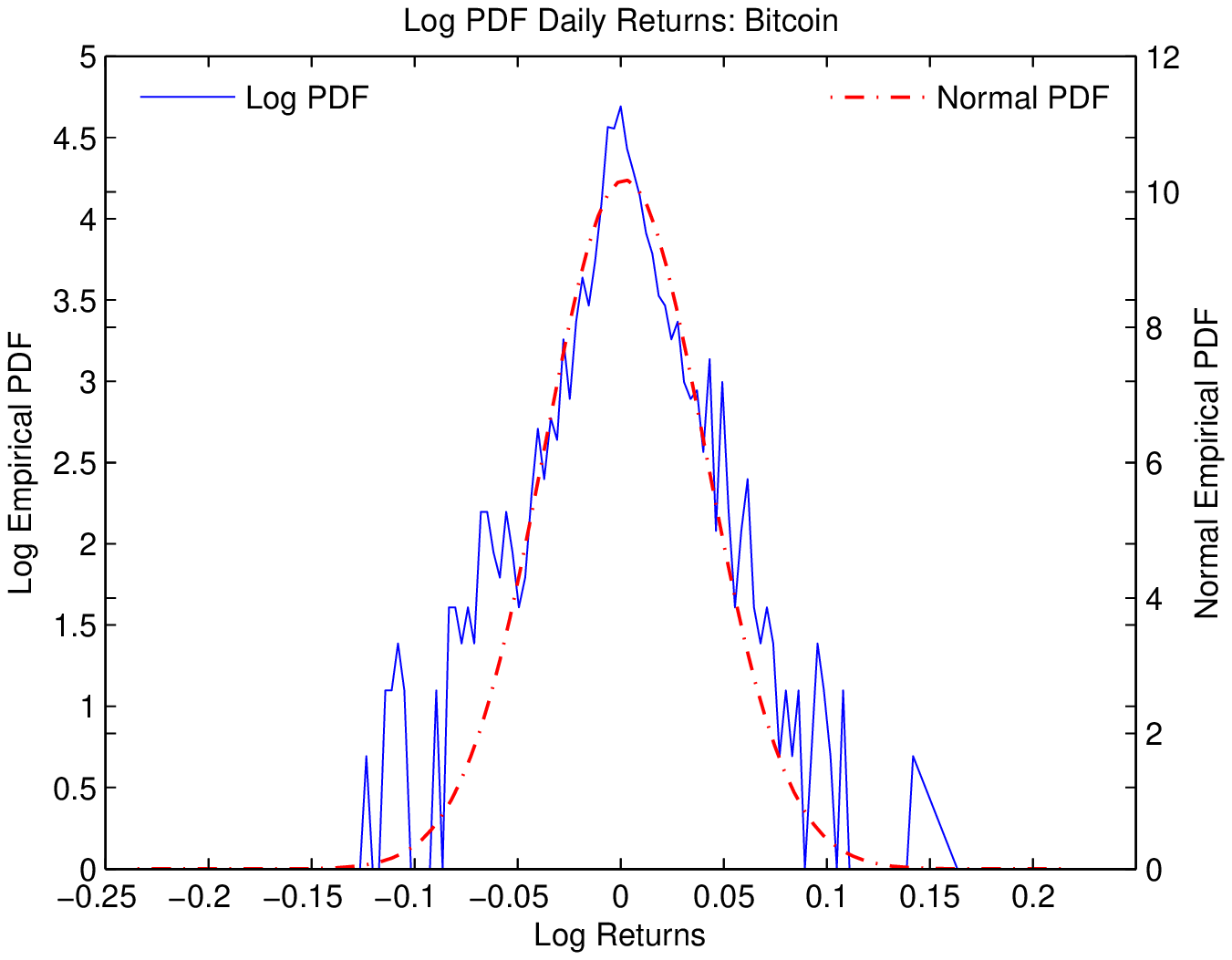}
\includegraphics[scale=.3]{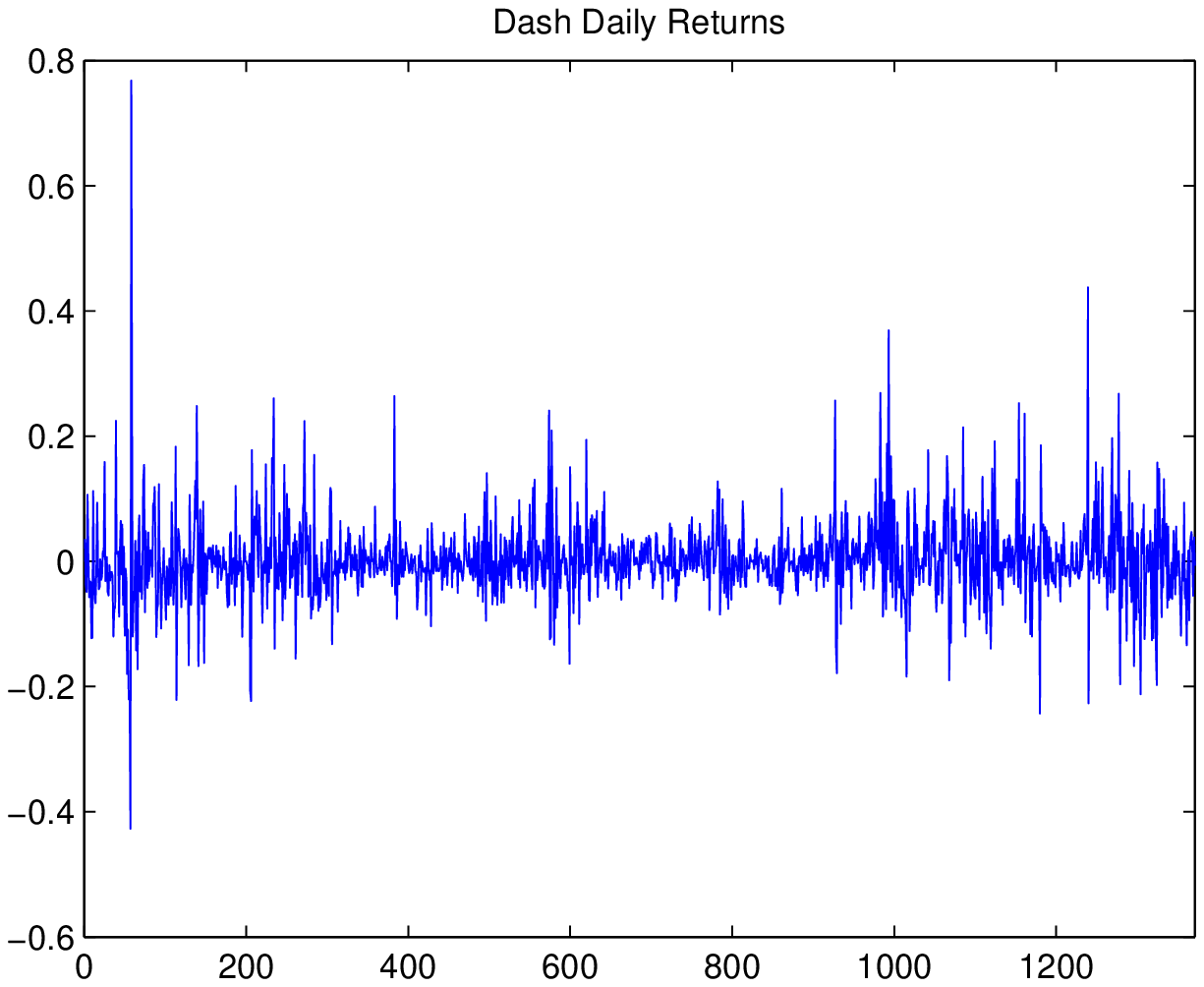}
\includegraphics[scale=.3]{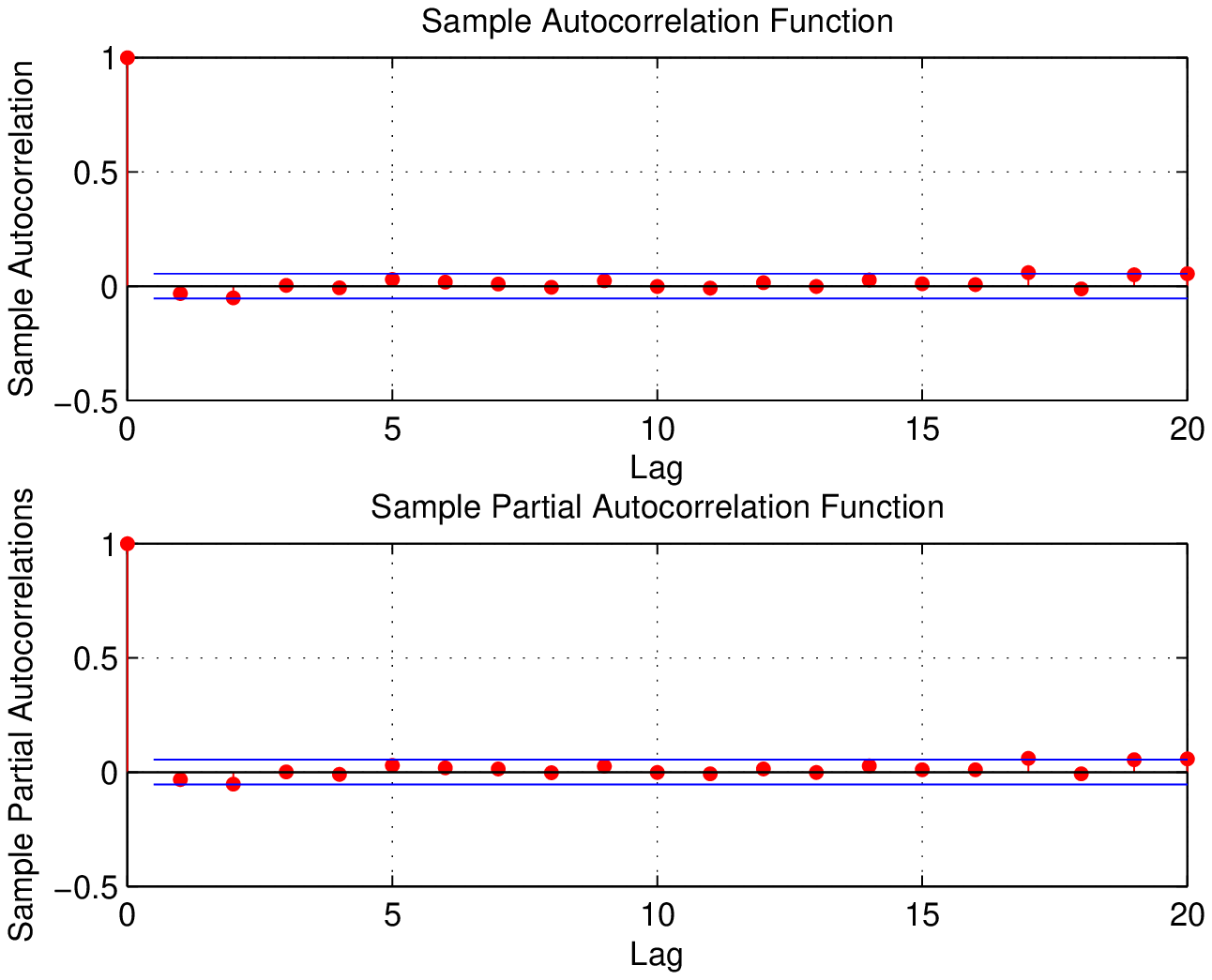}
\includegraphics[scale=.3]{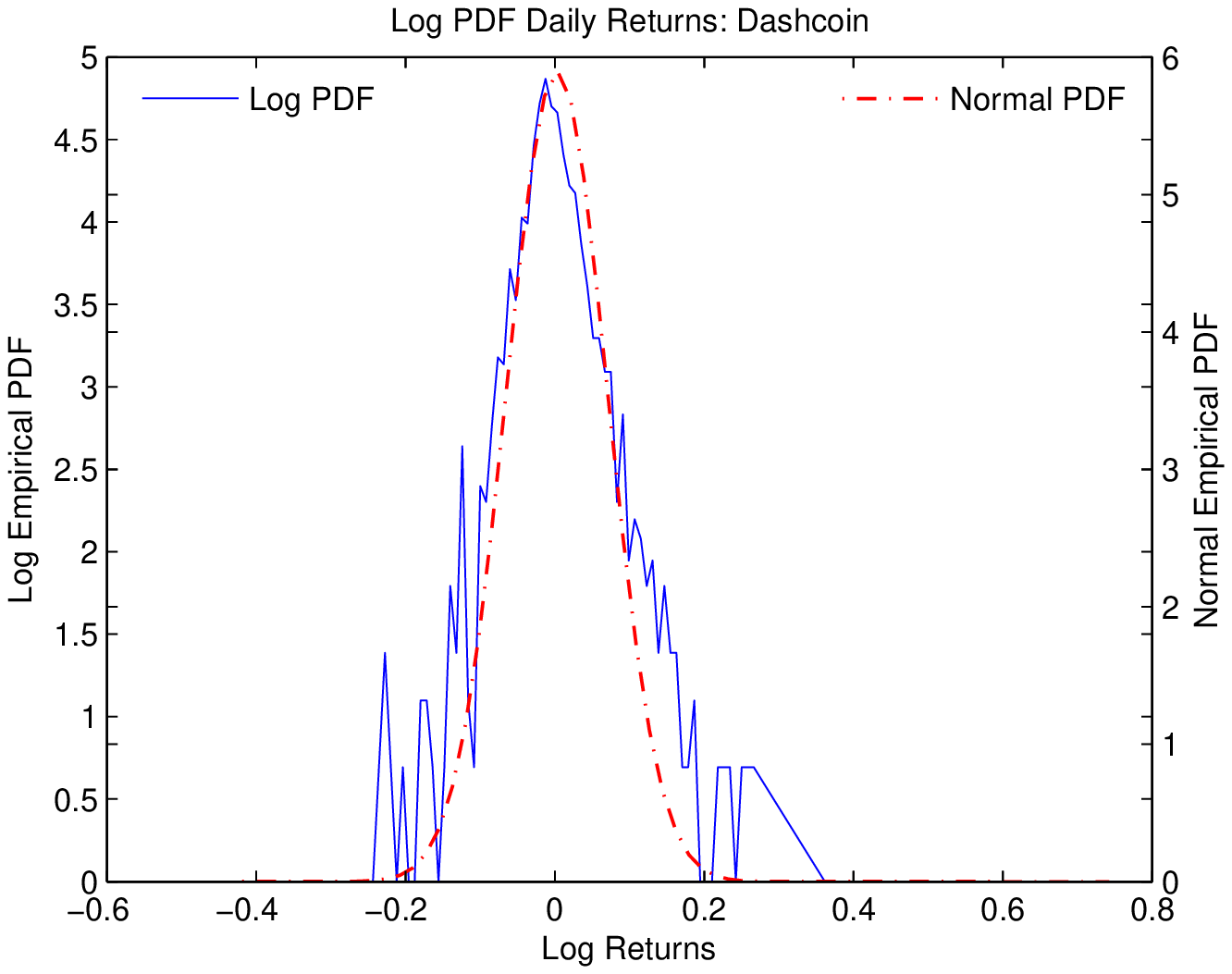}
\includegraphics[scale=.3]{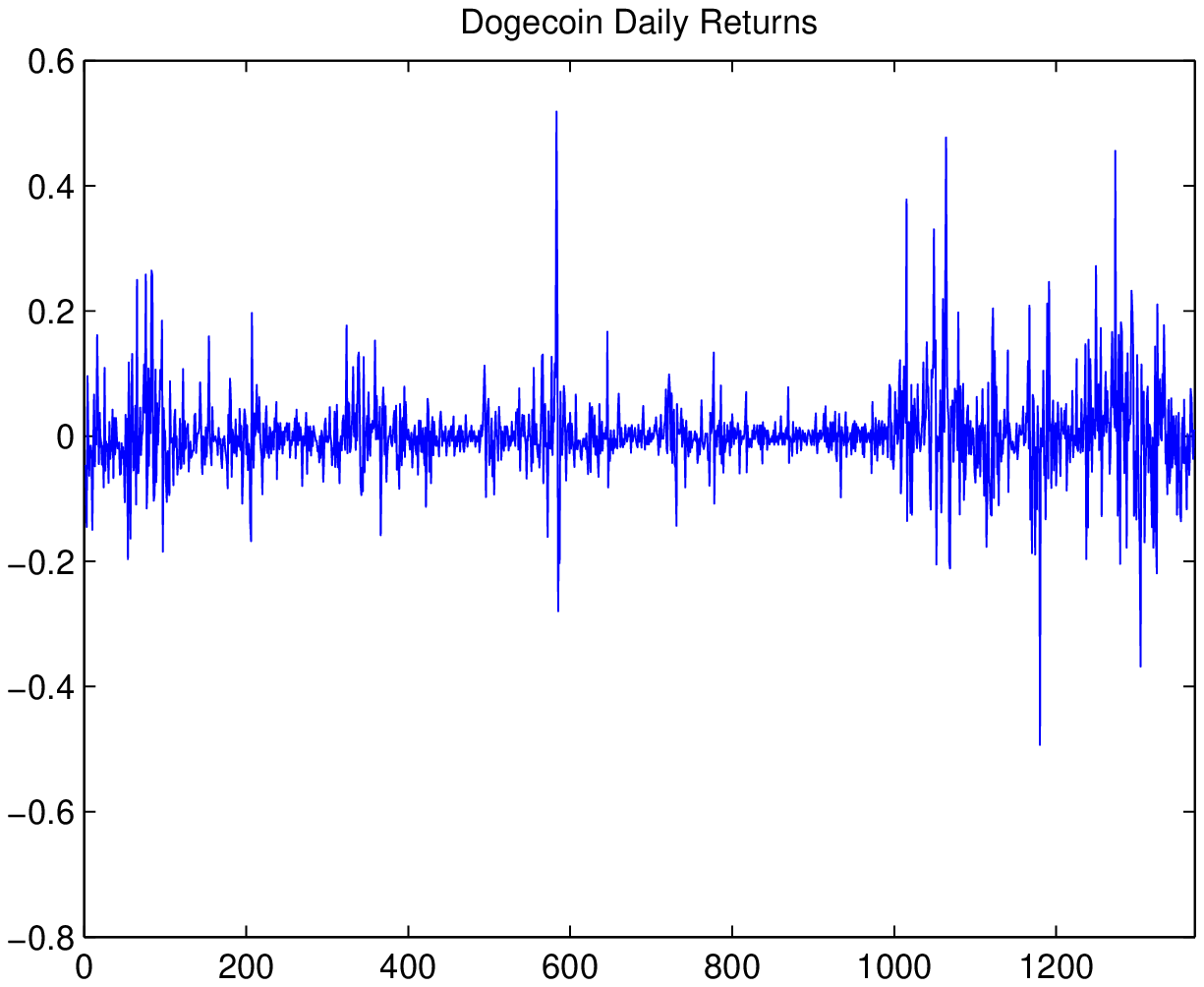}
\includegraphics[scale=.3]{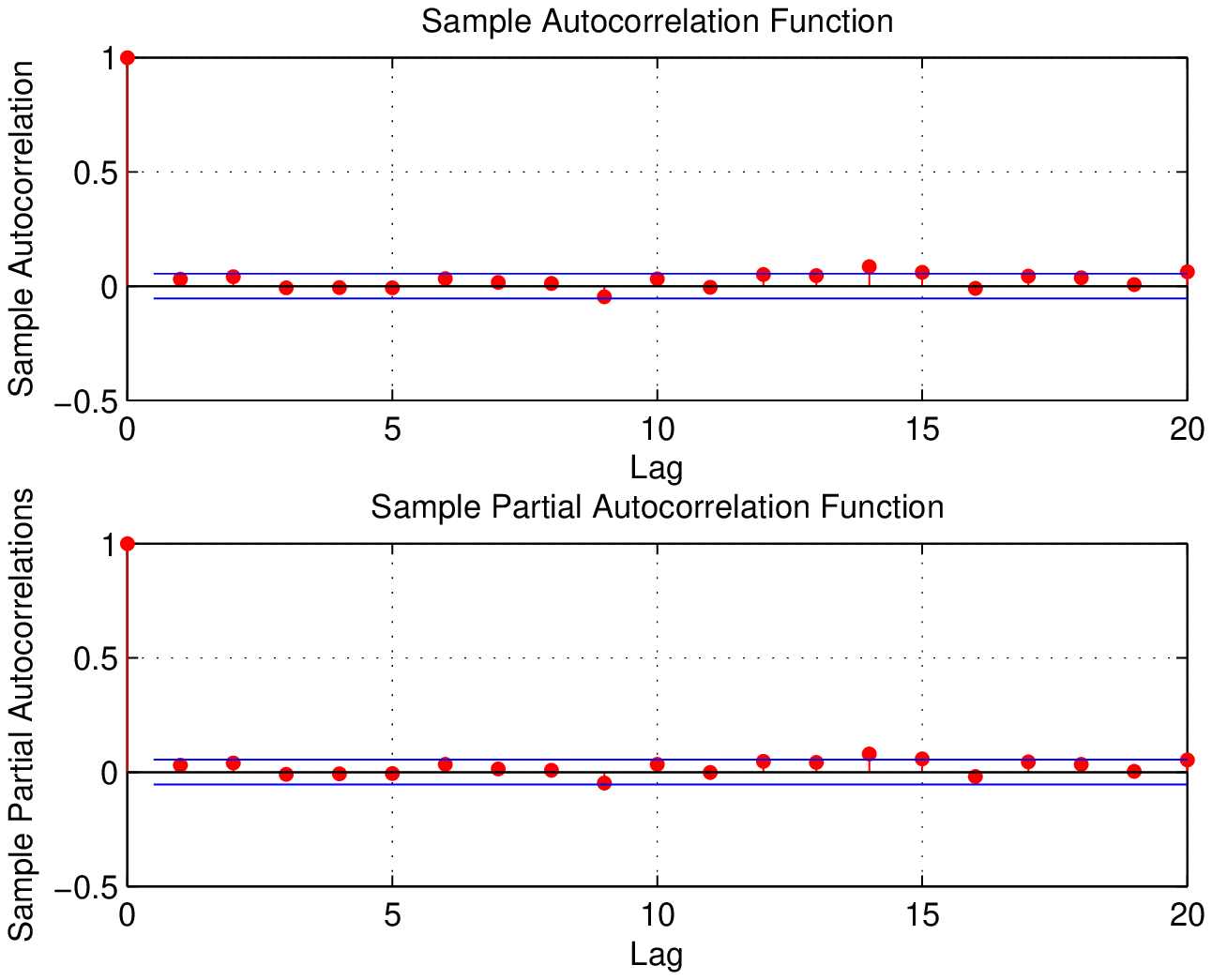}
\includegraphics[scale=.3]{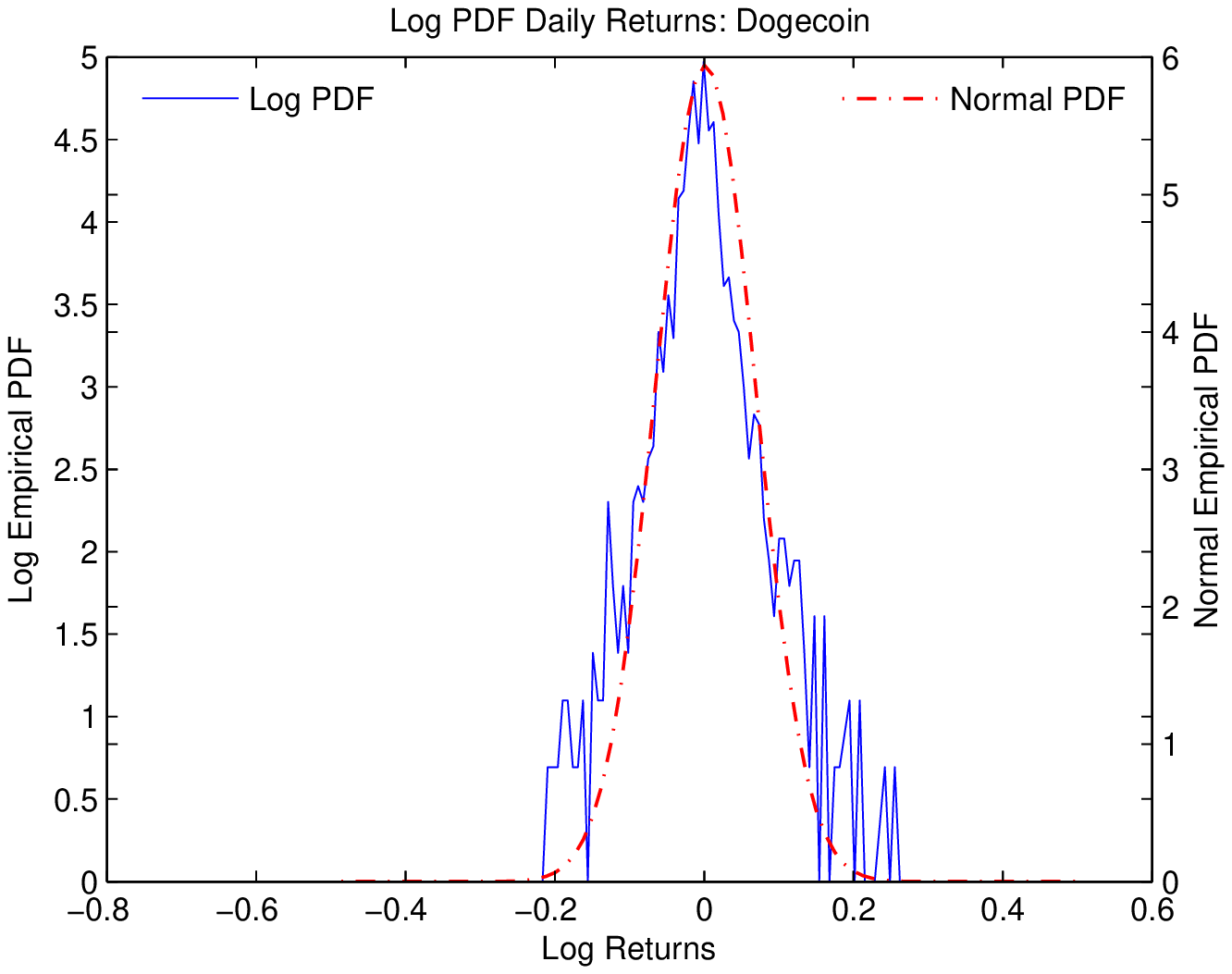}
\includegraphics[scale=.3]{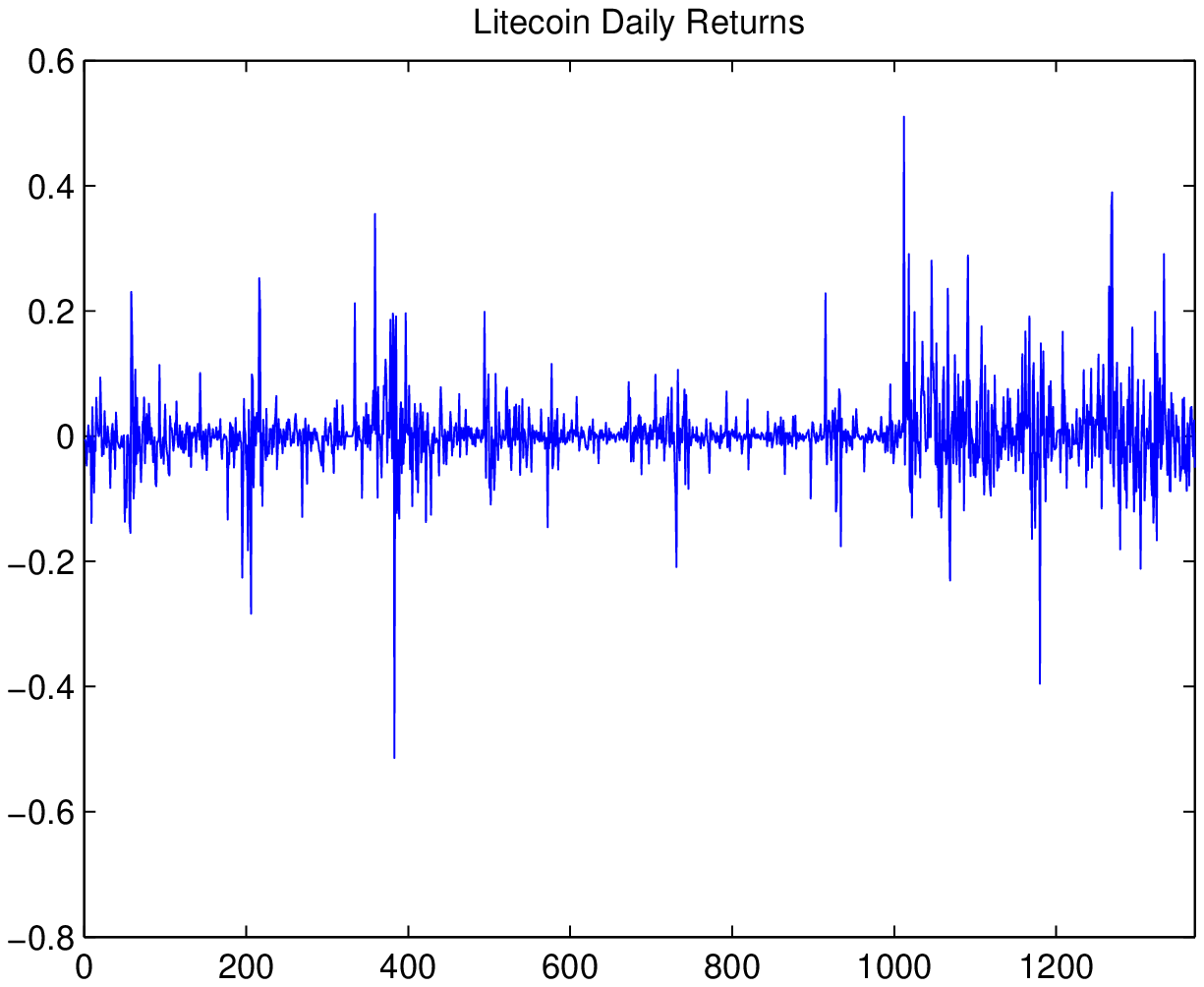}
\includegraphics[scale=.3]{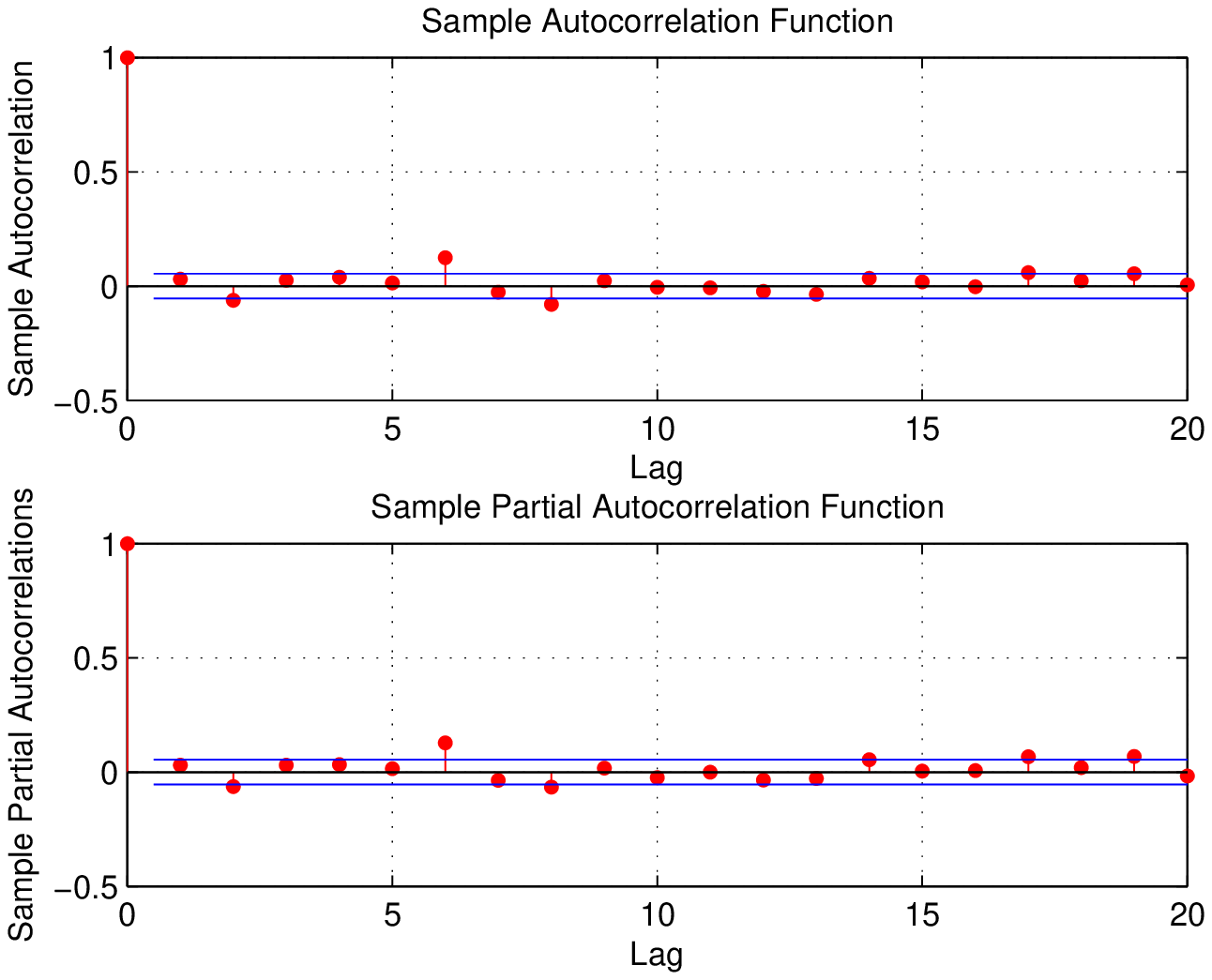}
\includegraphics[scale=.3]{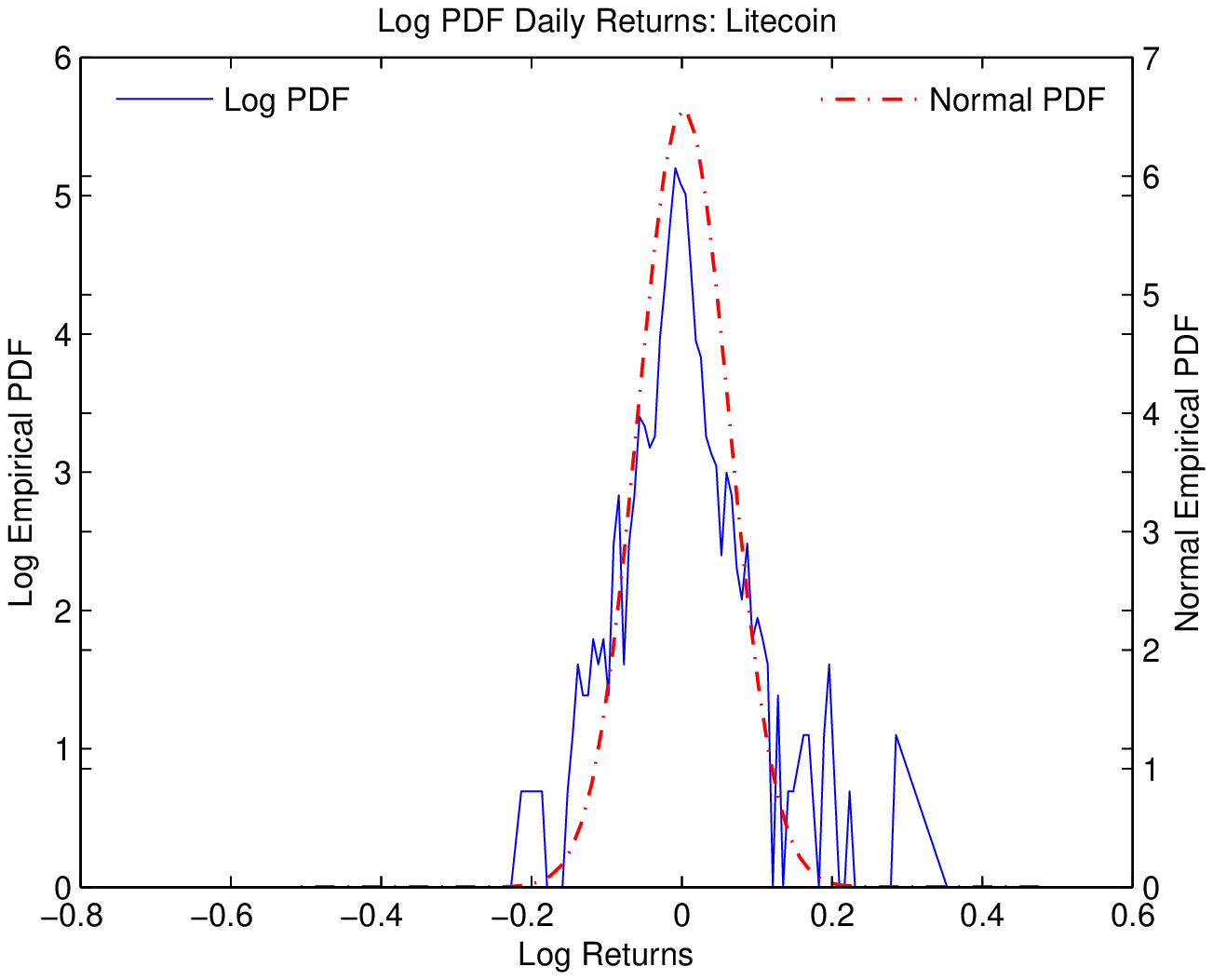}

\label{fig:r}
\end{figure}

\begin{figure}[htbp]
 \subsection{Daily log returns, ACF \& PCF, kurtosis in logarithmic returns - Maidsafecoin, Monero, and Ripple}
 \caption{(Data for Maidsafecoin, Monero, and Ripple from June 22, 2014 - March 24, 2018. (Left) Graphs of the daily log returns of the exchange rates. (Middle) Sample autocorrelation function (ACF) and partial autocorrelation function (PACF) for the return series. (Right) The solid blue line is the empirical Pdf of log returns. The dotted red line is the normal Pdf. Initial signs indicate presence of volatility clustering and kurtosis in logarithmic returns.} 
\centering
\includegraphics[scale=.3]{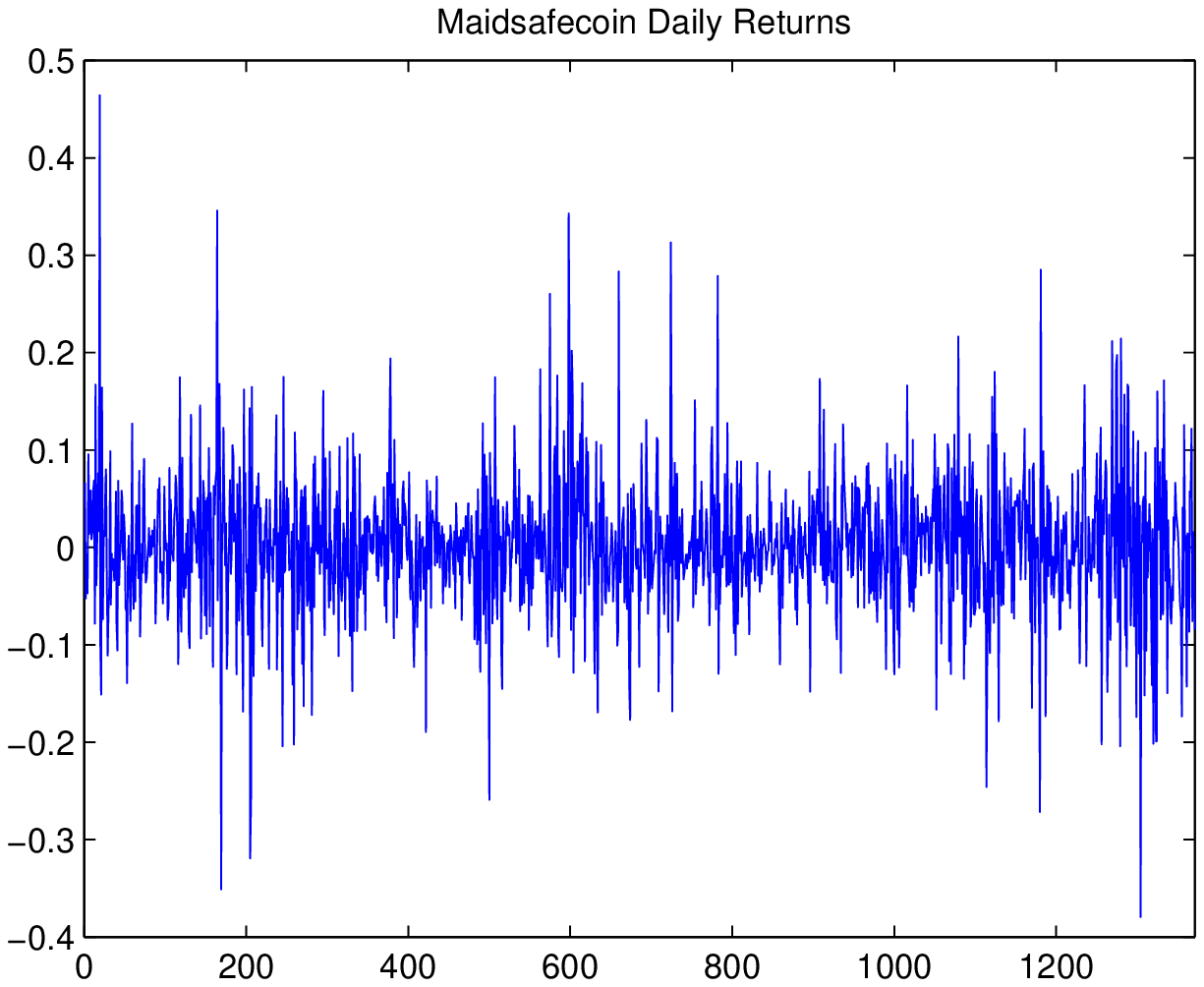}
\includegraphics[scale=.3]{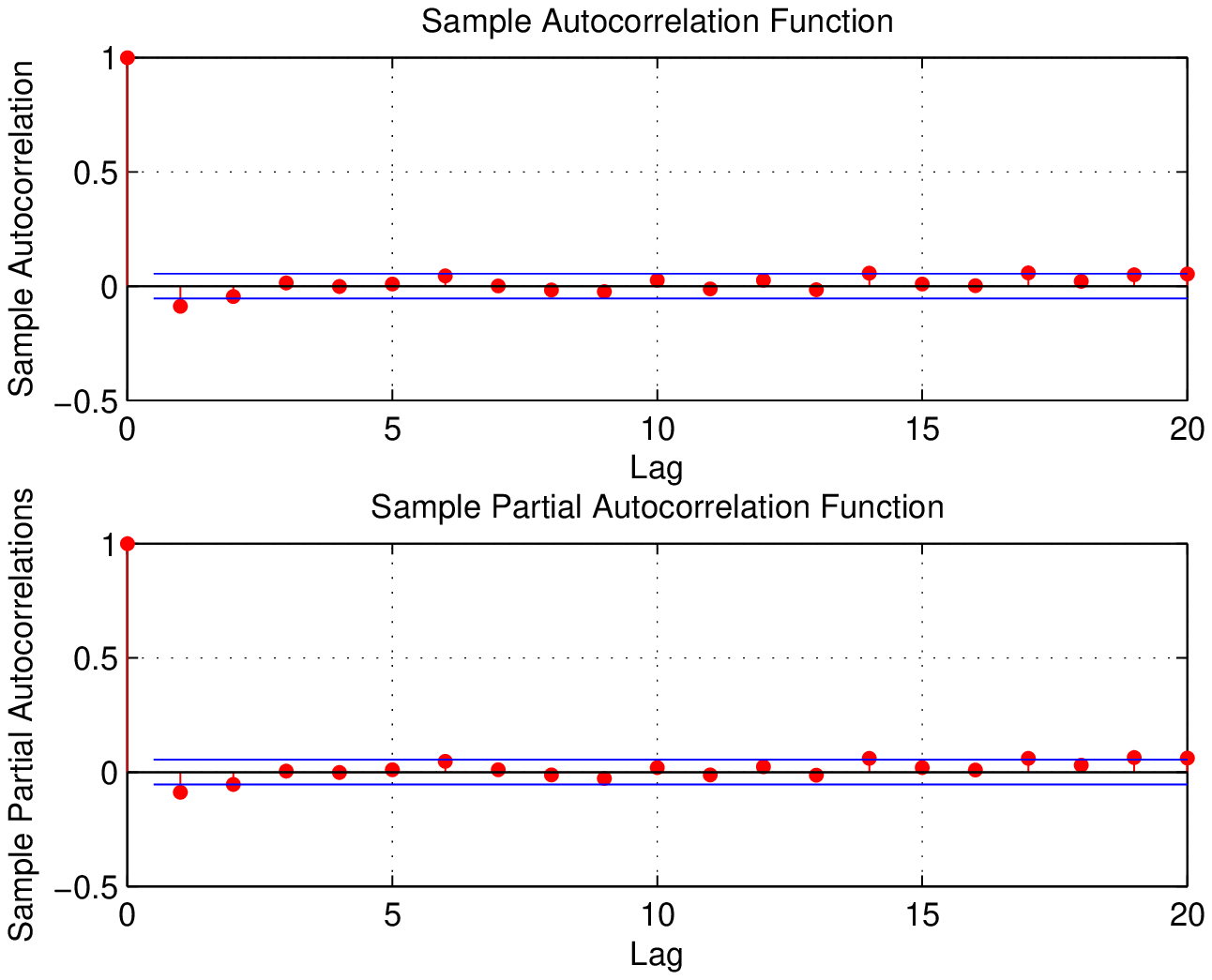}
\includegraphics[scale=.3]{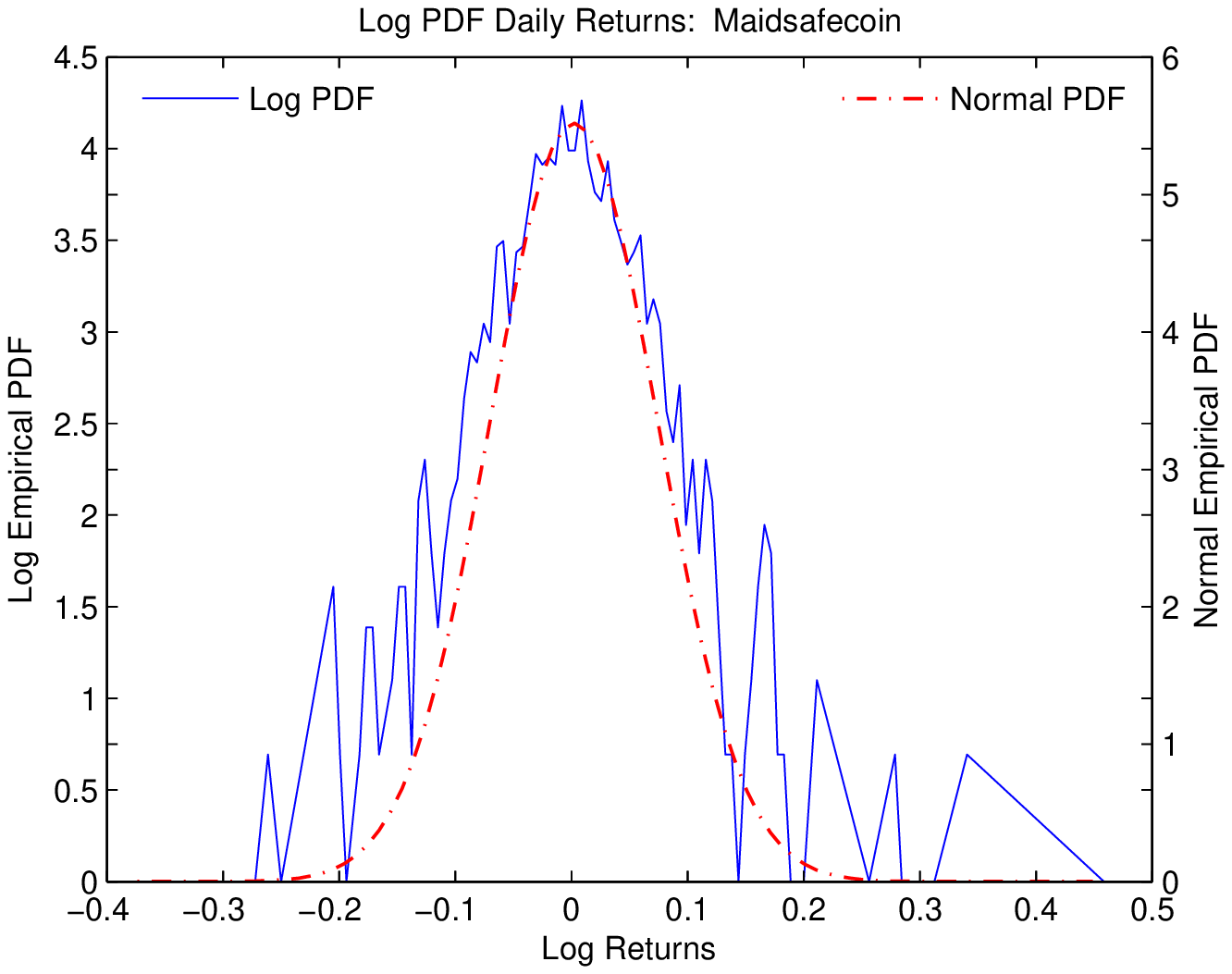}
\includegraphics[scale=.3]{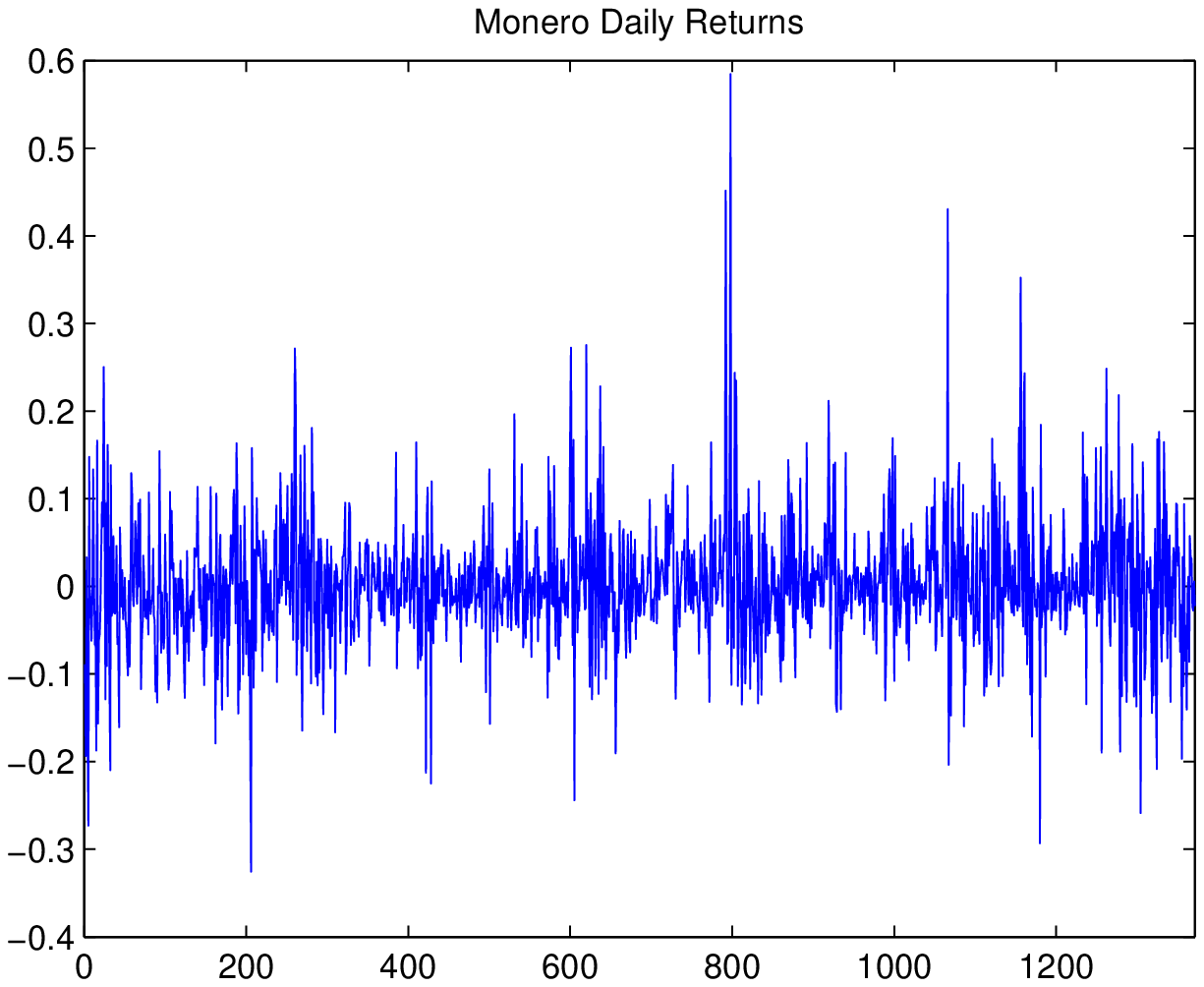}
\includegraphics[scale=.3]{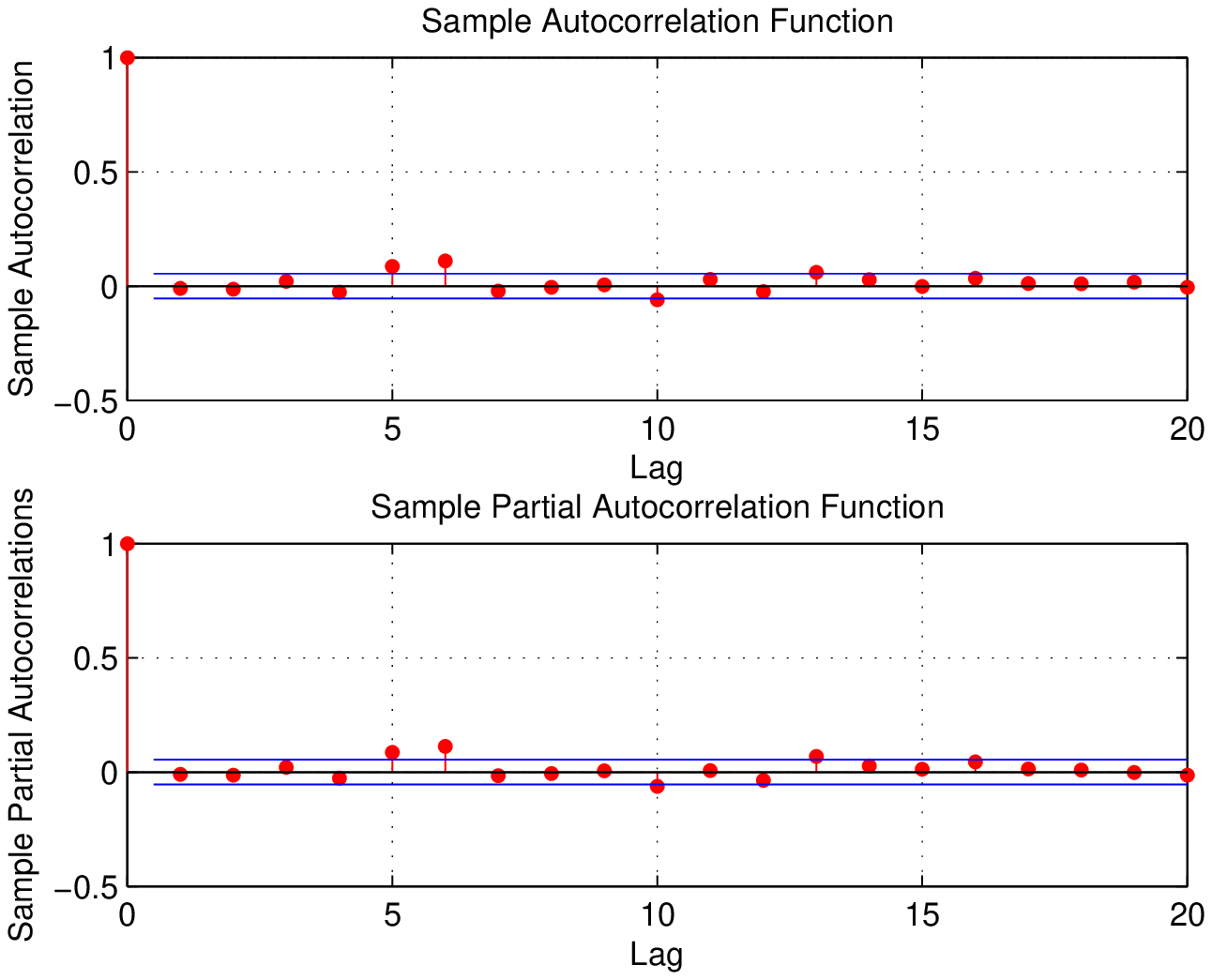}
\includegraphics[scale=.3]{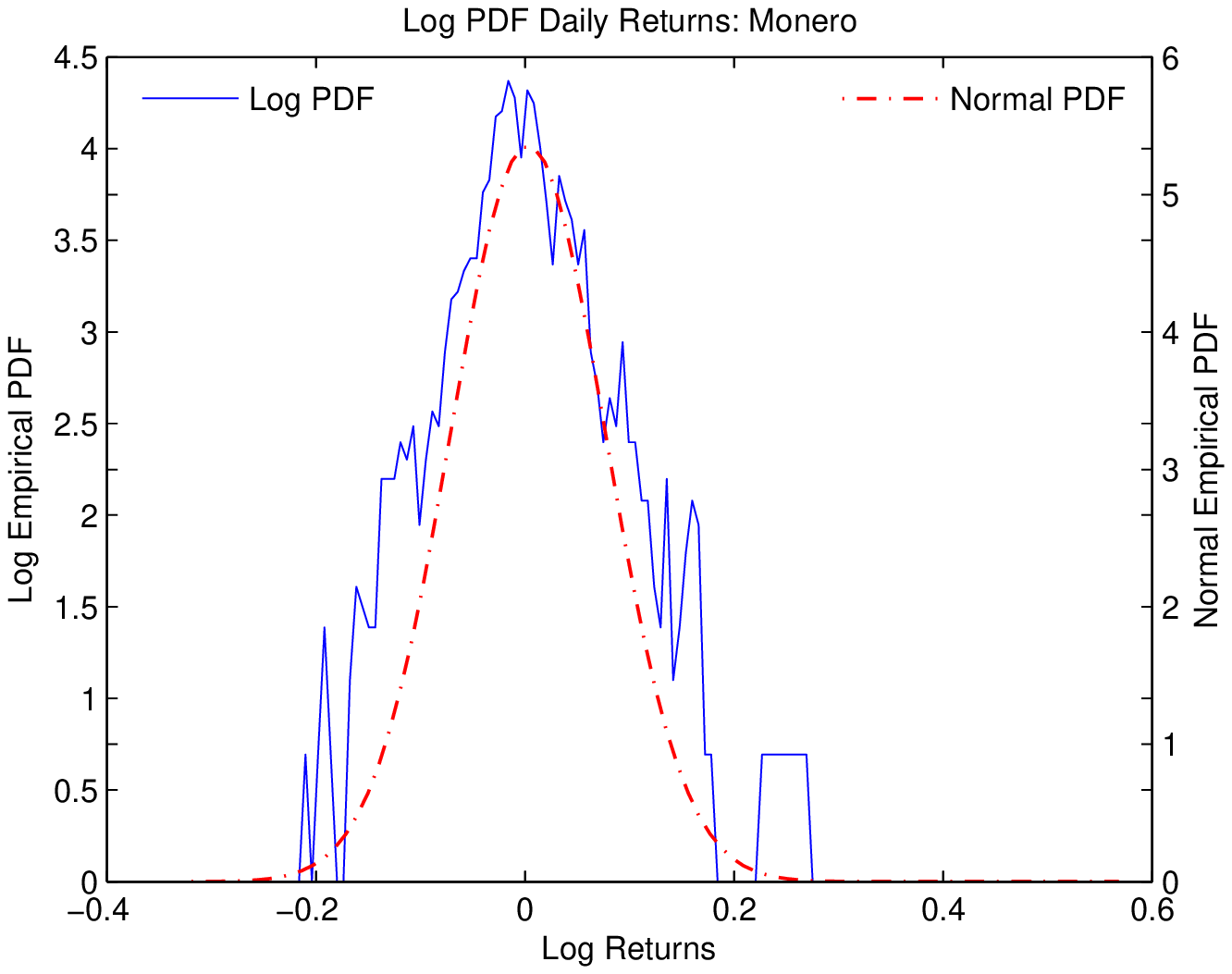}
\includegraphics[scale=.3]{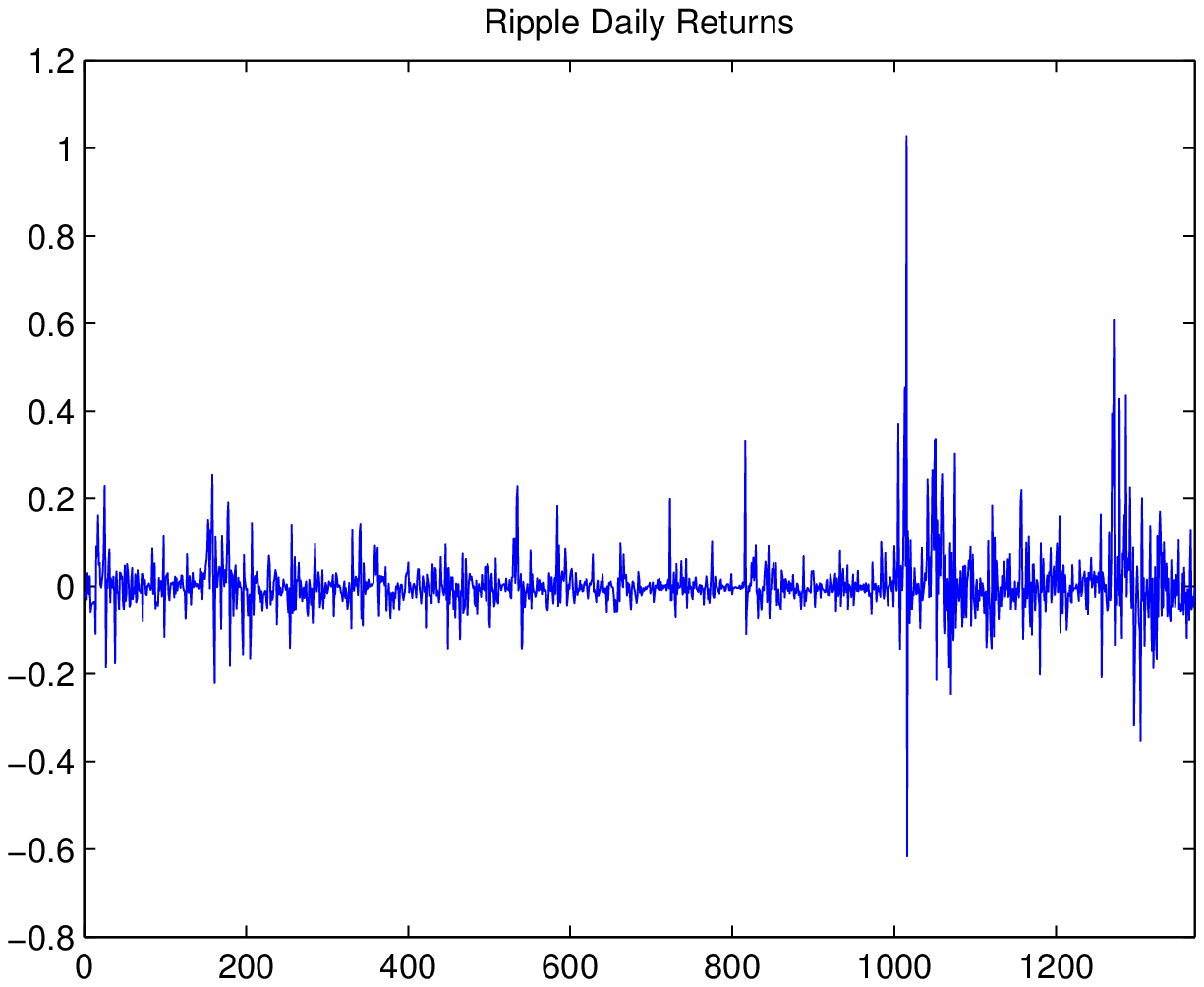}
\includegraphics[scale=.3]{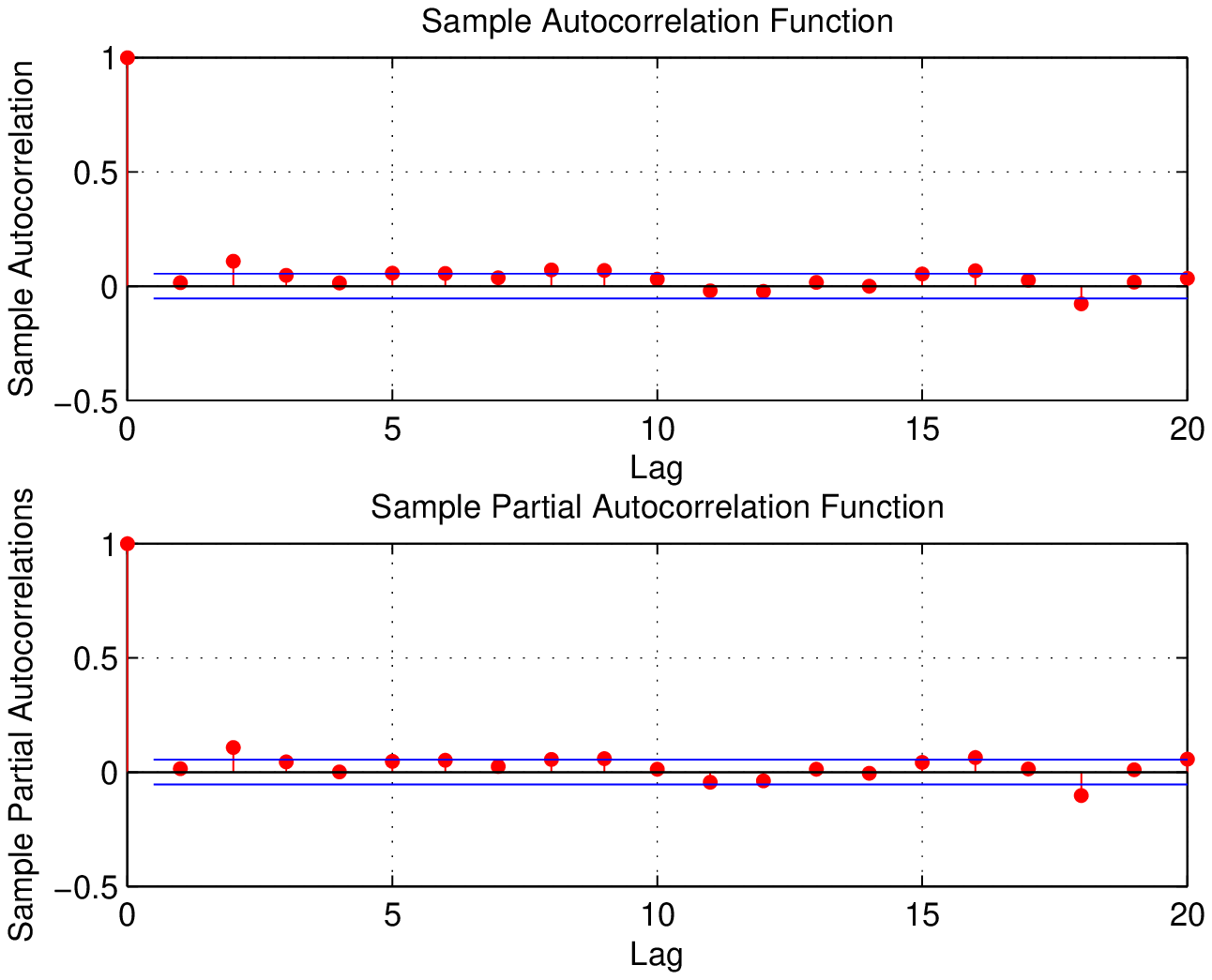}
\includegraphics[scale=.3]{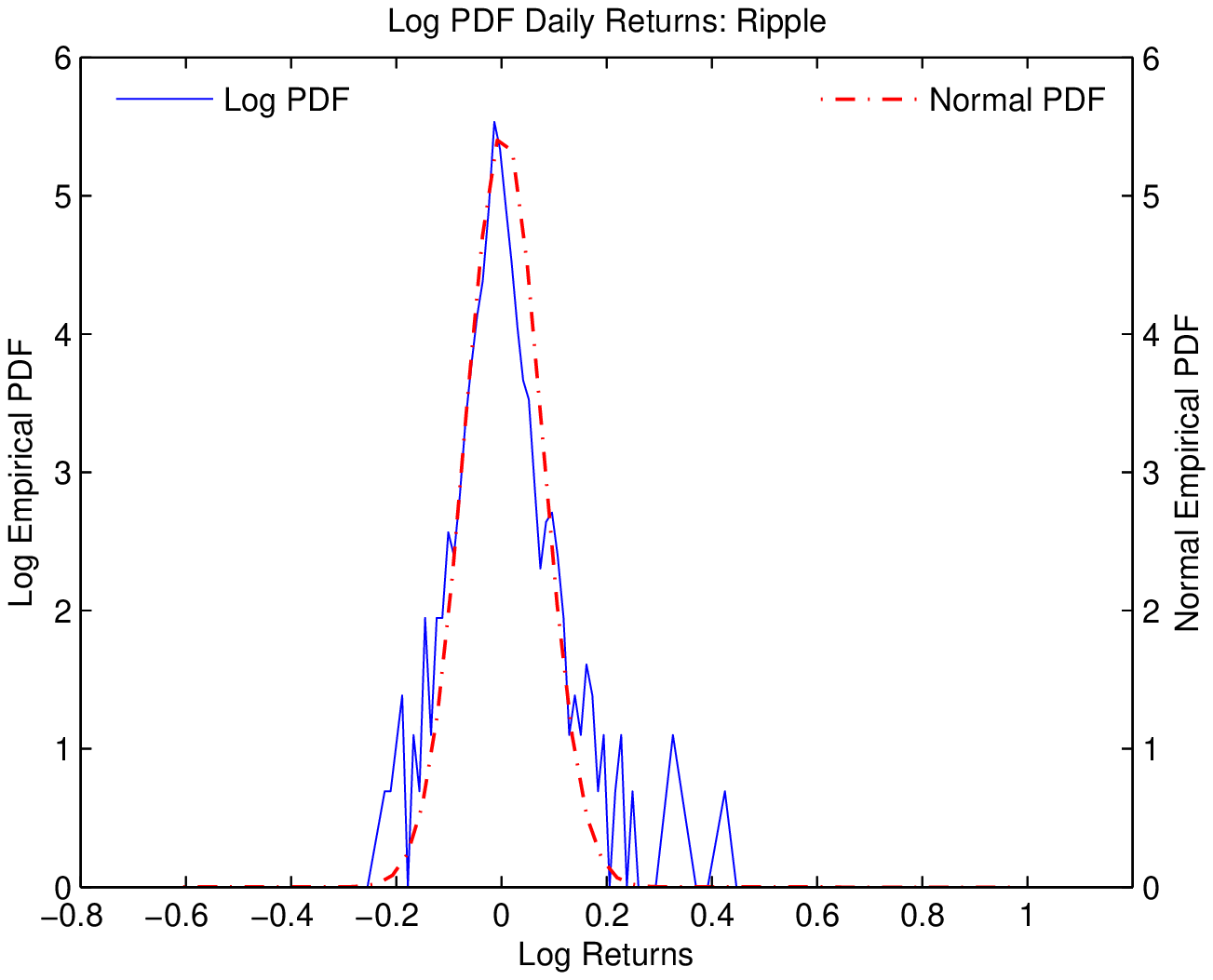}
\label{fig:acf}
\end{figure}

\newpage

\section{Goodness of fit}
This section demonstrates that, contrary to \cite{Chu17}, there is a strong empirical argument against modelling innovations under Gaussian assumptions. Further, arguing against \cite{Chu17}, one can also demonstrate a theoretical case for not relying on Skewed GED (SGED) assumptions, but using GED innovations instead. 

\subsection{Kolmogorov-Smirnov and Cramer von-Mises criteria}
Next, we outline the theoretical framework behind the goodness of fit tests, which are used for diagnosing the distribution of GARCH innovations. Many tests for normality exist in the literature. For example, the Jarque-Bera test is based on symmetry and kurtosis. The chi-square test is also widely used for distributional assumptions. However, these test have well-known issues. Jarque-Bera test tends to overreject the null hypothesis of normality in the presence of long memory in the series. Further, i.i.d. is the usual assumption in most of such tests. 

One alternative is to use the Kolmogorov-Smirnov type tests. However, it is  difficult to apply the Kolmogorov test in the presence of estimated parameters, particularly for multivariate data where the number of estimated parameters is large. If estimated parameters are ignored, the inference will be invalid. The method proposed by \cite{Bai03} addresses this problem by combining a Kolmogorov-Smirnov type test of conditional distribution specifications for time series with Khmaladze's K-transformation, as in \cite{K81}. The K-transformation takes the empirical distribution of pseudo-observations and maps it to a process $W$, that is asymptotically Brownian. Using an amended Bai method, outlined in \cite{Remi13}, we now check whether innovations are Gaussian \footnote{ These checks are not carried out in \cite{Chu17}}. To this end, two appropriate tests are Kolmogorov-Smirnov and Cramer-von-Mises criteria. Both tests are used since KS is a general test and can be under-powered.

\noindent The distribution function $F$ of the Kolmogorov-Smirnov criterion is 
\[  F(x)=\frac{4}{\pi} \sum^\infty_{k=1} \frac{(-1)^k}{2k+1}e^{-(2k+1)^2\pi^2/(8x^2)}.\]
The distribution of the Cramer von-Mises criterion is the same as
\[  \frac{4}{\pi^2} \sum^\infty_{k=1} \frac{Z^2_k}{2k+1},\]
where $Z_1,Z_2,\ldots$ are $i.i.d.$ standard Gaussian.

\begin{table}[!htbp]
\centering
\caption{Quantiles for Kolmogorov-Smirnov and Cramer von-Mises criteria.}
\label{Table1}
\begin{tabular}{|l|l|l|}
\hline	
 Confidence level    & Kolmogorov-Smirnov  &  Cramer von-Mises   \\ \hline
90\% & 	1.96			   &  1.2				 \\
95\% & 	2.241			   &  1.657				 \\
99\% &	2.807			   &  2.8		 \\ \hline		 
\end{tabular}
\end{table}

\subsection{Khmaladze's martingale transformation }
The argument is presented as follows. Define $e_i$, which are GARCH(p,q) pseudo-observations, as $e_i=\frac{x_i-\hat{\mu}_i}{\hat{\sigma}_i}$, $i\in \{1,\ldots,n\}$. Let $u_i=N(e_i)$, $i\in \{1,\ldots,n\}$. The associated order statistics are $v_1,\ldots,v_n$. Set conditions $v_0=0$, $v_{n+1}=1$. Following \cite{Bai03}, we define
\[  \dot{g}(s)=\begin{bmatrix} 1 \\
-N^{-1}(s) \\
[-N^{-1}(s)]^2
\end{bmatrix},
\textnormal{~and~}
C(s)=\int^1_0 \dot{g}(t)  \dot{g}^{\top}(t) dt, \quad s \in (0,1).
\]
If $a=-N^{-1}(s)$ and $$x=-N^{'}(a)=\frac{e^{-a^2/2}}{\sqrt[]{2\pi}},$$ then 
\[ C(s)=
\begin{bmatrix}
1-s & -x & -ax \\
-x &  1-s+ax  & x(1+a^2)  \\
-ax & x(1+a^2) & 2(1-s)+ax(1+a^2) 
\end{bmatrix} \]
for all $s\in (0,1)$. Then define $V_n(s)=\frac{1}{\sqrt[]{n}} \sum^n_{i=1} \{I(v_i \leq s) -s\}$. The K-transform of $V_n$ is 
\[ 
W_n(s)=V_n(s)-\int^s_0 \left\{\dot{g}^\top(t) C^{-1} (t)  \int^1_t \dot{g}(\tau) dV_n(\tau)\right\} dt, \quad s \in [0,1]
\]
$W_n$ is approximately Brownian under the null of innovations being standard Gaussian $N(0,1)$. Then, \cite{Bai03} suggests to approximate $C(v_j)$ by 
\[ 
\sum^n_{k=j} (v_{k+1}-v_k) \dot{g}(v_k) \dot{g}^\top(v_k),
\]
and this can be calculated exactly for Gaussian innovations. On the other hand, computing 
\begin{equation}
\label{quadra}
\int^{v_k}_{v_{k-1}} C^{-1} (t)\dot{g}(t) dt
\end{equation}
is difficult, which is perhaps the reason why \cite{Chu17} avoid this methodology. \cite{Bai03} suggests to approximate it by
\[C^{-1}(v_{k-1}) \int^{v_k}_{v_{k-1}} \dot{g}(t) dt= C^{-1}(v_k) \{\{g(v_k)-g(v_{k-1})\}.
\]
Instead, we follow the method proposed by \cite{Remi13} and estimate \eqref{quadra} using Gauss-Kronrod quadrature, such that 
\[
\int^{v_k}_{v_{k-1}} C^{-1} (t)\dot{g}(t) dt \approx (v_k-v_{k-1}) C^{-1} 
\Big( \frac{v_{k-1}+v_k}{2} \dot{g} \frac{v_{k-1}+v_k}{2}
\Big).\]
Only then the Kolmogorov-Smirnov and Cramer von-Mises test criteria are used to assess goodness of fit.
\[ \textrm{Kolmogorov-Smirnov} := ~KS= \max_{j\in \{1,\ldots,n \}} |W_n(v_j)| , ~~ 
\textrm{Cramer von-Mises} := ~ CvM= \frac{1}{n} \sum_{j=1}^n W^2_n (v_j) (v_{j+1}-v_j) \] 

\subsection{Generalised Error Distribution}
The Generalised Error Distribution (GED) is an alternative to the Gaussian, which has some attractive properties that are naturally amenable to modelling innovations in our context. A r.v. $X$  of parameter $\nu>0$, is $X \sim GED(\nu)$ if its density is 
\[ 
f_v(x)= \frac{1}{b_v 2^{1+1/v}\Gamma \Big(1+1+\frac{1}{\nu}\Big)} e^{-\frac{1}{2}\frac{|x|}{b_\nu}}, \quad x \in \R,
\textnormal{~~where~~}
b_\nu=2^{-\frac{1}{\nu}}\sqrt[]{\frac{\Gamma\Big( \frac{1}{v} \Big)}{\Gamma\Big( \frac{3}{v} \Big)}  }.
\]
Let $F_\alpha$ be the distributed gamma with parameters $\alpha=\beta=1$. Then $F$ of $X\sim GED (\nu)$, and its inverse $F^{-1}$, are 
\begin{equation*}
F(x)=\left\{
                \begin{array}{ll}
                 \frac{1}{2}- \frac{1}{2} F_{1/\nu}\Big( \frac{1}{2} \Big(\frac{|x|}{b_\nu} \Big)^\nu \Big) & \quad  x\leq 0\\
                 \frac{1}{2}+ \frac{1}{2} F_{1/\nu}\Big( \frac{1}{2} \Big(\frac{x}{b_\nu} \Big)^\nu \Big)              & \quad  x>0
                \end{array}
              \right.
\textnormal{,~~and~}
F^{-1}(x)=\left\{
                \begin{array}{ll}
                 \frac{1}{2}- \frac{1}{2} F_{1/\nu}\Big( \frac{1}{2} \Big(\frac{|x|}{b_\nu} \Big)^\nu \Big) & \quad  0<u\leq \frac{1}{2}\\
                 \frac{1}{2}+ \frac{1}{2} F_{1/\nu}\Big( \frac{1}{2} \Big(\frac{x}{b_\nu} \Big)^\nu \Big)              & \quad  \frac{1}{2}\leq u < 1.
                \end{array}
              \right.
\end{equation*}

\subsubsection{Statistical limitations of SGED vs GED: when mgf fails to exist}
This section presents two key results, which are offered in support of arguments presented in the introduction section of this paper. First, it is shown why the moment generating function (mgf) of a Generalized Error Distribution (GED) exists when $v\geq 1$ and fails to exist when $0<v<1$. Second, it is shown why the mgf of the Skewed GED (SGED) fails to exist for any $k\neq 0$, an important set of conditions for estimation. These arguments proceed as follows. 

Let $M(t)=Ee^{tX},-\infty<t<\infty$ denote a mgf. The pdf of GED is
\[f_v(x)=\frac{v ~exp (-\frac{1}{2}|\frac{x}{\lambda}| v }{\lambda 2 ^{1+1/v} \Gamma\Big( \frac{1}{v} \Big)}, \quad v>0, x \in \R   \]
where $\Gamma(\cdot)$ is the Gamma function, and 
\[\lambda= \Biggl[ 2^{-\frac{2}{v}}\frac{\Gamma\Big( \frac{1}{v} \Big)}{\Gamma\Big( \frac{3}{v} \Big)}  \Biggr]^{\frac{1}{2}}.
\]
The pdf of SGED is
\[g(x)= \frac{exp \Big(-\frac{1}{2} \Big(-\frac{1}{k} \ln\Big(1-\frac{k(x-\eta)}{\alpha} \Big)\Big)^2\Big)}{\sqrt[]{2\pi}\alpha \Big(1-\frac{k(x-\eta)}{\alpha}\Big)}\]
where $x \in \Big( -\infty, \eta+\frac{\alpha}{k}\Big)$ if $k>0$, $x\in \Big(\eta+\frac{\alpha}{k}, \infty\Big)$ if $k<0$. Further, $\eta$ and $\alpha$ are a real constant and a positive constant respectively. When $k \to 0$, $g(\cdot)$ reduces to the pdf of a random normal variable (r.v.) with mean $\eta$ and variance $\alpha^2$.

\begin{lemma}
\label{GED1} Let a r.v. $X$ be distributed GED s.t. $v>0$. Then the moment-generating function M(t) exists $\forall$ $t$, when $v>1$; $\exists$ in the region $(-\sqrt[]{2},\sqrt[]{2})$ when $v=1$; and does not exist $\forall$ $t>0$ when $0<v<1.$
\end{lemma} 

\begin{proof}
\label{GED1p} Let us take any $v>0$ so that $M(t)=c_1 \int^\infty_{-\infty} e^{tx} e^{-c_2|x|^v}dx, \quad -\infty<t<\infty$, where $c_1=v \Big(\lambda 2^{1+1/v} \Gamma\Big(\frac{1}{v}\Big)\Big)^{-1}$ and $c_2=(2 \lambda ^v)^{-1}$. Suppose that $0<v<1.$ Then $\forall x>0$
\[e^{tx} e^{-c_2|x|^v}= e^{tx\Big(1- \frac{c_2}{tx^{1-v}}  \Big)}.\]
Let $x_0>0$ be such that $\frac{c_2}{tx^{1-v}}<\frac{1}{2}$ $\forall$ $x \geq x_0$ so that 
\[e^{tx \Big(1- \frac{c_2}{tx^{1-v}}  \Big)} \geq e^{tx/2} \quad \forall x \geq x_0. \]
It follows that 
\[ \int^\infty_{-\infty} e^{tx} e^{-c_2|x|^v} dx=\infty. \]
Therefore $M(t)$ $\nexists$ $\forall$ $t>0$ when $0<v<1$. When $v=1$, the pdf $f(x)=\frac{1}{\sqrt[]{2}}e^{-\sqrt[]{2}|x|}$, $-\infty<x<\infty.$ Let 
\[ M(t)= \frac{1}{\sqrt[]{2}} \int^\infty_{-\infty} e^{tx-\sqrt[]{2}|x|}  dx = \frac{1}{\sqrt[]{2}}(I_1+I_2),\]
where 
\[
I_1=\int^0_{-\infty} e^{tx+\sqrt[]{2}x} dx 
\quad \textrm{and} \quad 
I_2=\int^{-\infty}_{0} e^{tx-\sqrt[]{2}x} dx.\]
Letting $x=-y$ yields
\[
I_1=\frac{1}{t+\sqrt[]{2}} 
\quad \textrm{and} \quad 
I_2=\frac{1}{t-\sqrt[]{2}}
\]
Then,
\[M(t)=\frac{1}{\sqrt[]{2}} 
\Big(
\frac{1}{\sqrt[]{2}+t} +  \frac{1}{\sqrt[]{2}-t}
\Big)=2\Big(\Big(2-t^2\Big)\Big)^{-1}=\Big(1-\frac{t^2}{2} \Big)^{-1}
\]
for any $|t|<\sqrt[]{2}$. Further, $\nexists$ $M(t)$ when $|t|\geq \sqrt[]{2}$.

Let $v>1$. Then define
\begin{eqnarray*}
M(t)&=&c_1 \int^\infty_{-\infty} e^{tx-c_2|x|^v}dx, \quad -\infty<t<\infty  \\
&=& c_1 \left\{ \int^0_{-\infty} e^{tx-c_2|x|^v}dx + \int^\infty_{0} e^{tx-c_2|x|^v}dx
\right\} \\
&=& c_1 (I_1+I_2).
\end{eqnarray*}
Since $tx-c_2x^v
=-c_2x\Big( 1-\frac{t}{c_2x^{v-1}} \Big)
=-c_2x^v\Big( 1+o(1) \Big)$, it is evident that as $x \to \infty$, $I_2<\infty$ $\forall$ $t_{t \in \R}$.
Let $x<0$. Then, 
\[tx-c_2|x|^v
=-c_2c_2|x|^v \Big(1-\frac{tx}{c_2|x|^v} \Big) 
=-c_2c_2|x|^v \Big(1-\frac{t}{c_2|x|^{v-1}} \Big) 
=-c_2c_2|x|^v\Big(1+o(1)\Big) \textrm{~as~} x\to - \infty
.\]
Therefore $I_1<\infty$ $\forall$ $t$. Then $M(t)$ exists for all $t\in (-\infty, \infty).$ Apply Maclaurin expansion and letting $M^k(0)=EX^k, k\geq 0$ yields
\[M(t)=\sum^\infty_{k=0}\frac{t^k}{k!} EX^k.  \]
By symmetry ($X$ is symmettric around 0), we get $EX^k=0$ for when $k$ is odd. For when $k$ is event, $k=2m$, 
\[EX^{2m}=\Big(\frac{\Gamma(\frac{1}{v})}{\Gamma(\frac{3}{v})} \Big)^m  
\frac{\Gamma(\frac{2m+1}{v})}{\Gamma(\frac{1}{v})}, \quad m \geq 1. \]
Therefore 
\[M(t)=\sum^\infty_{m=0} \frac{t^{2m}}{(2m)!}\Big(\frac{\Gamma(\frac{1}{v})}{\Gamma(\frac{3}{v})} \Big)^m  
\frac{\Gamma(\frac{2m+1}{v})}{\Gamma(\frac{1}{v})}.
\]
Hence the result. A closed form expression is not available for $M(T)$.  
\end{proof}

\begin{lemma}
\label{GED2} 
The mgf $M(t)$ of a SGED does not exist for any $t>0$ when $k<0$, and for any $t<0$ when $k>0$. 
\end{lemma}

\begin{proof}
\label{GED2p} 
Case 1: $k<0$. If $k<0$, then the pdf of SGEDS is  
\[g(x)= 
\frac{exp \Big( -\frac{1}{2} \Big(-\frac{1}{k} log\Big(1+\frac{-k(x-\eta)}{\alpha} \Big)\Big)^2\Big)}{\sqrt[]{2\pi}\alpha \Big(1+\frac{-k(x-\eta)}{\alpha}\Big)}, \quad x\geq \eta+\frac{\alpha}{k},\]
where $\alpha >0, \eta \in (-\infty,\infty)$ are constants.

The mgf is \[M(t)=\sum^\infty_{\eta+\frac{\alpha}{k}} 
e ^{tx} g(x) dx, \quad -\infty <t< \infty. \]
Let $1+(-k)\frac{x-\eta}{\alpha}=y$. then $x=\alpha(-k)^{-1}(y-1)+\eta$ and therefore
\begin{eqnarray}
\label{GED2.1} 
m(t) &=& \frac{e^{(\eta+\alpha/k)t}}{\sqrt[]{2\pi}\alpha} \int^\infty_0 \Big( e^{-\frac{\alpha t y}{k}-\frac{(\log y)^2}{2k^2}-\log y }\Big) dy.
\end{eqnarray}
Then 
\[  
-\frac{\alpha t y}{k}-\frac{(\log y)^2}{2k^2}-\log y=-\frac{\alpha y t}{k} \Big( 1+ \frac{k \Big( \frac{(\log y)^2}{2k^2}+\log y\Big)}{\alpha y t}.
\]
Therefore, for any $t>0$, there is a $y_0$ s.t. 
\[ 
-\frac{\alpha t y}{k}-\frac{(\log y)^2}{2k^2}-\log y\geq - \frac{\alpha y t}{2k}. 
\]
Using the property of $\eqref{GED2.1}$, we have 
\[M(t) \geq \frac{e^{(\eta+\alpha/k)t}}{\sqrt[]{2\pi}(-k)} \int^\infty_{y_0} \Big(e^{(\alpha y t)/(2k)}\Big) dy=\infty, \]
therefore the mgf does not exist $\forall k<0$, just as it does not exist $\forall$ $t>0$.\\

Case 2: $k>0$. Let $X$ be SGED with $k>0$. Define $Y=-X$. The pdf of $Y$ is 
\[  
h(y)=\frac{exp\Big(-\frac{1}{2}\Big(\frac{-1}{(-k')}\log \Big(1+\frac{(-k')(y-\eta')}{\alpha}\Big)\Big)^2\Big)}{\sqrt[]{2\pi}\alpha \Big(1+\frac{(-k')(y-\eta')}{\alpha}\Big)}, \quad y \geq \eta'+\frac{\alpha}{k'}, \]
where $\eta'=-\eta$ and $k'=-k$. Then $Y$ is SGED with $k'<0$, and therefore $Ee^{tY}$ does not exist for any $t>0$. Therefore, $M(t) = Ee^{tX}$ $\nexists$ $\forall$ $t<0$, as required. 
\end{proof}

%%%%%%%%%%%%%%%%%%%%%%%%%%%%%%%%%%%%%%
\section{Empirical Results}
\subsection{Testing if innovations are Gaussian.}
After estimating the parameters of the GARCH(1,1) model for Bitcoin data with constant mean and Gaussian innovations, we proceed to test the hypothesis of whether the innovations are gaussian. Although a test of normality was proposed in \cite{Bai03}, we make use of an updated implementation that deploys a Gauss-Kronrod quadrature. 

Both the KS and CvM goodness of fit tests reject strongly (at 1\%) the null of gaussianity of innovations for GARCH(1,1). For GARCH(1,1) using Bitcoin data, the K-S test statistic was 10.02, and the CVM test statistic was 13.58. In fact, the null of gaussianity is rejected strongly for all GARCH(p,q) models, with $p,q \in \{1,\ldots,5 \}$ (see Table \ref{Tablep}). Just by way of illustration, Figure \ref{fig:trajectory} plots the Brownian motion paths of the innovations process trajectories for all GARCH(p,q) models, with $p,q \in \{1,\ldots,5 \}$ for Bitcoin. The dotted line indicates critical values for a 95\% level for the Kolmogorov-Smirnov statistic. Repeating this exercise for the rest of the currencies in sample produces comparable results. Table \ref{Tablep} shows the P-Values of KS for GARCH(1,1) across all currencies in sample. The null of gaussianity is strongly rejected for all GARCH(p,q) models, with $p,q \in \{1,\ldots,5 \}$, for all cryptocurrencies in sample. 

\begin{figure}[!htb]
\caption{Brownian motion path of a innovations process trajectories for all GARCH(p,q) models, with $p,q \in \{1,\ldots,5 \}$, using (left to right) Bitcoin data. The results are presented without loss of generality since results for the other cryptoassets considered in this paper are directly comparable, see Table \ref{Tablep}. The dotted line indicates 95\% confidence level for Kolmogorov-Smirnov statistic.}
\centering
\includegraphics[scale=.7]{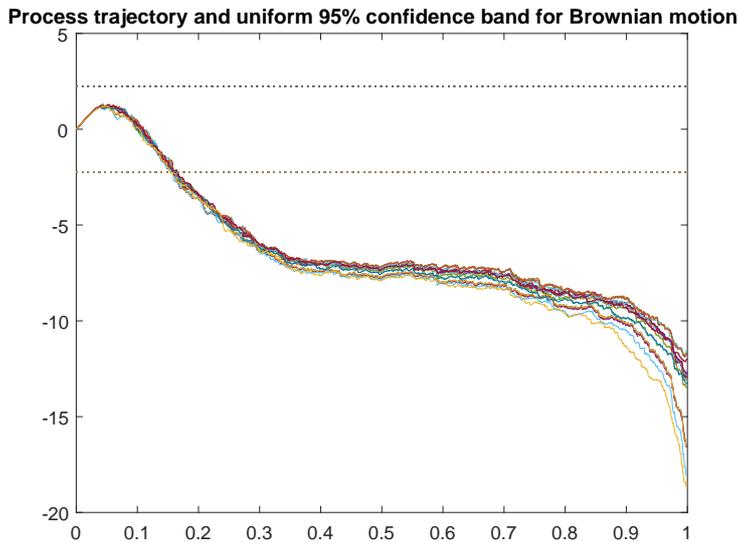}
\label{fig:trajectory}
\end{figure}

\begin{table}[!htbp]
\centering
\caption{KS test, null of gaussianity. }
\begin{tabular}{|l|l|l|}
\hline	
\label{Tablep}
Asset   & Reject at 95\%? & Reject at 99\%?\\
\hline		 
Bitcoin  &   Yes & Yes\\
Dash &  Yes & Yes\\
Dogecoin    &  Yes & Yes\\
Litecoin    & Yes & Yes\\
Maidsafecoin    &  Yes & Yes\\
Monero	&  Yes & Yes\\
Ripple & Yes & Yes\\

\hline		 
\end{tabular}
\end{table}

\newpage
\subsection{Dealing with non-Gaussian innovations}

Next, using maximum likelihood, parameters of the GARCH(1,1) model are estimated with constant mean and Generalized Error Distribution (GED) innovations.  Following \cite{Remi13}, we apply the Khmaladze transform for GED innovations to obtain pseudo-observations $u_{n,i}=G_{\hat{\nu}}(e_i),i \in \{1, \ldots, n\}$. The tests are based on the empirical distribution function 
\[
D_n(u)=\frac{1}{n}\sum^n_{i=1} \mathbb{I} (u_{n,i}\leq u), ~~u \in [0,1]. \] This should approximate the uniform distribution function $D(u)=u$ for $u\in [0,1]$ under the null that innovations follow GED distribution.

For calculating P-values, the parametric bootstrap method is used as per \cite{GK2014}. To calculate the bootstrap statistics, for the models and sample sizes considered, $N=1000$ bootstrap samples were used. Using Bitcoin data, for GARCH(1,1) the K-S test statistic was 0.9265 (p-value 8.4 \%), and the CVM test statistic was 0.1956 (p-value 5.8\%). Both the goodness of fit tests fail to reject the null of GED innovations for GARCH(1,1) using Bitcoin data. Further, we fail to reject the null of GED innovations for all GARCH(1,1) models for all cryptocurrencies in sample.

The distributions are now plotted in order to visually be able to compare specific aspects for differences. For illustration, Figure \ref{fig:Ged95} shows that the empirical process $D_n$ lies within the 95\% confidence band for the currencies in sample. The bootstrap algorithm ($N=1000$) takes around 1hr to run for all currencies using MatlabR2017a, on a 64-bit pc with 4gb of RAM.

%Litecoin data, for GARCH(1,1) :\\
%Cramer-von Mises statistic: 0.5222, P-Value (\%): 2.10\\
%Kolmogorov-Smirnov statistic: 1.7841, P-Value (\%): 0.10.\\

%Dashcoin data, for GARCH(1,1) :\\
%Cramer-von Mises statistic: 0.2568, P-Value (\%): 1.50\\
%Kolmogorov-Smirnov statistic: 1.1149, P-Value (\%): 1.00

\begin{figure}[!htb]

\caption{Approximating the uniform distribution function $D(u)=u$ for $u\in [0,1]$ under the null that innovations follow GED distribution. The empirical process $D_n$ lies within the 95\% confidence band: (left to right) Bitcoin,	Dash, Dogecoin,	Litecoin, Maidsafecoin, Monero, and Ripple. In fact, we fail to reject the null of GED innovations for all GARCH(1,1) models for all cryptocurrencies in sample. See 'Dealing with non-Gaussian innovations' subsection.}
\centering
\includegraphics[scale=.3]{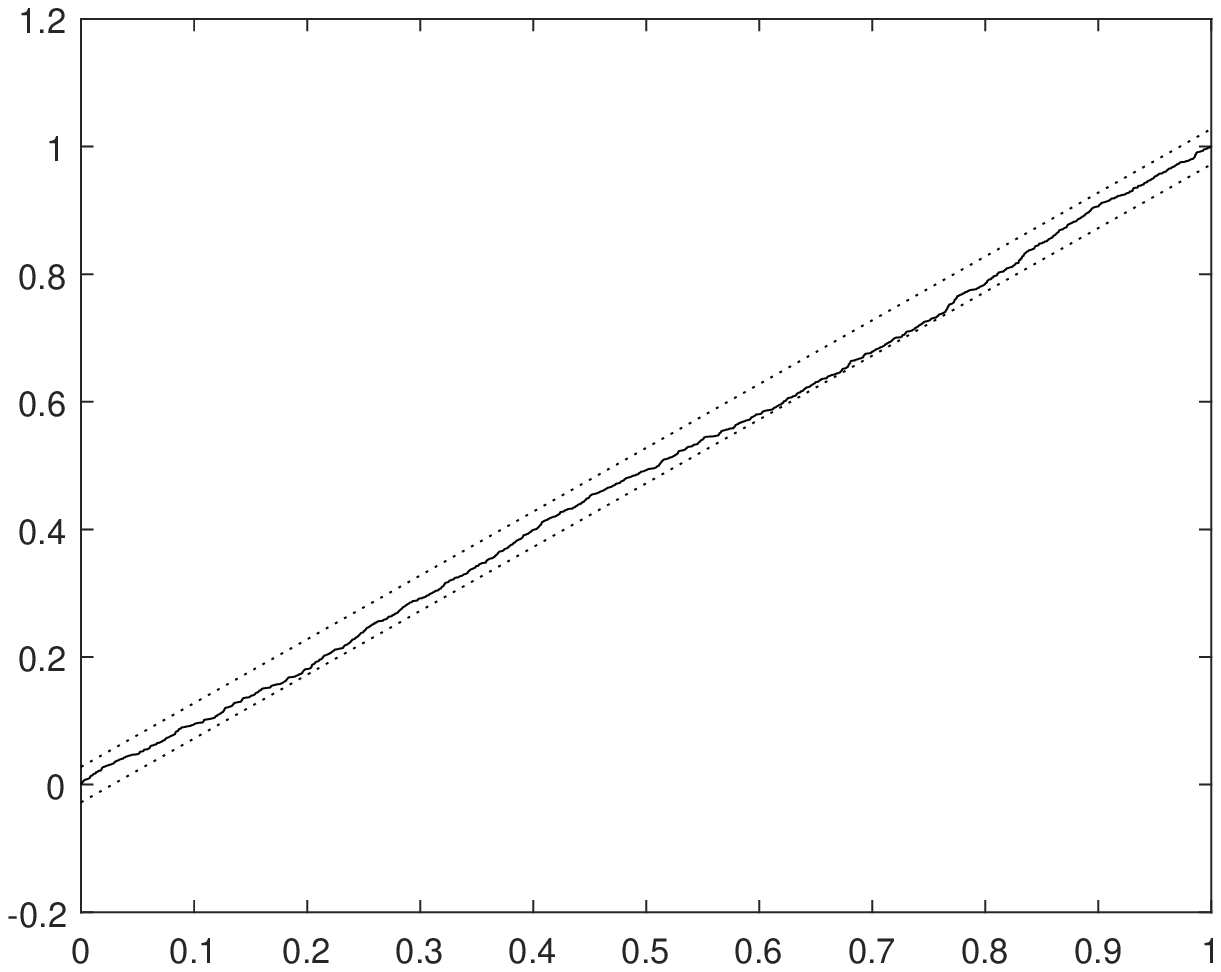}
\includegraphics[scale=.3]{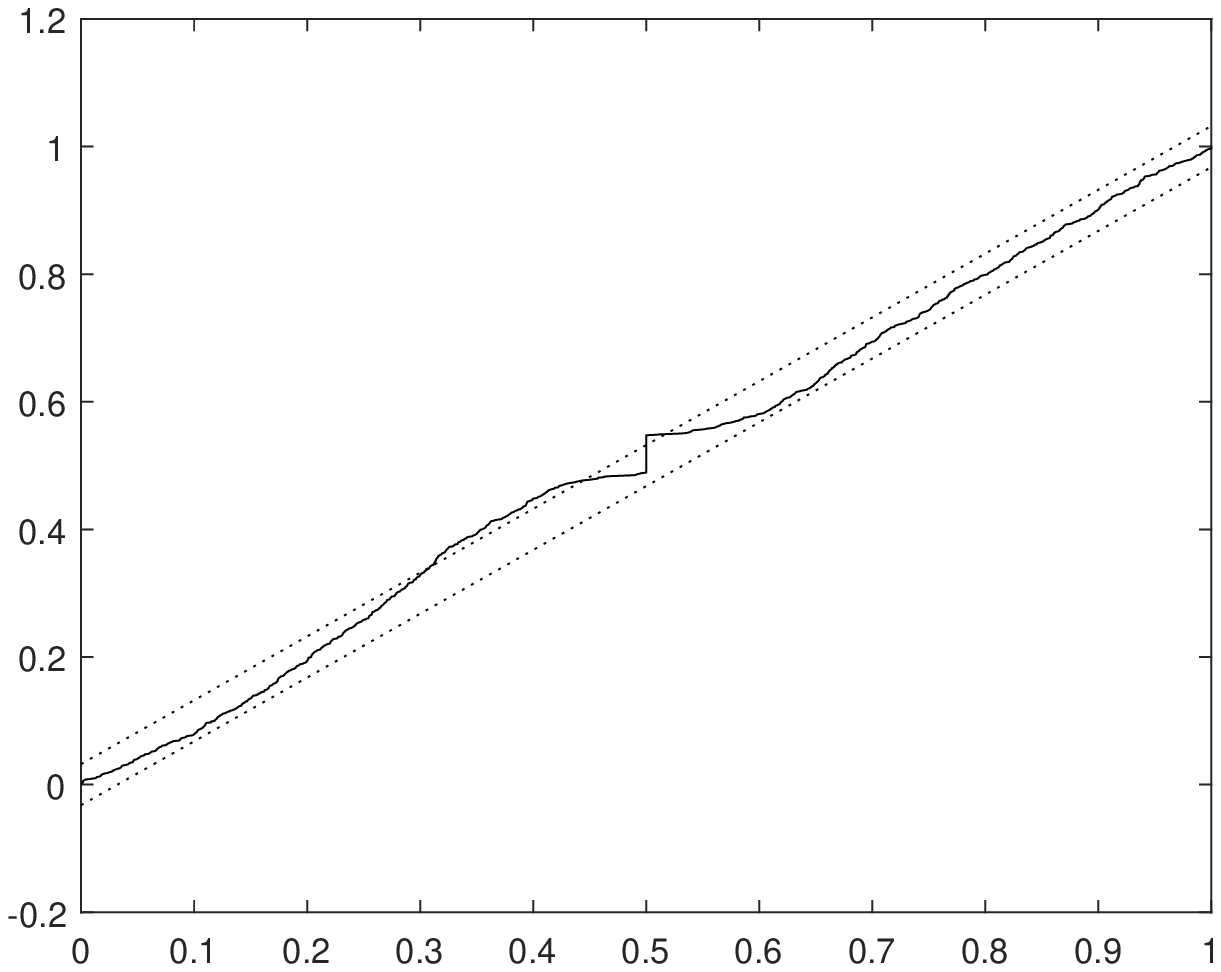}
\includegraphics[scale=.3]{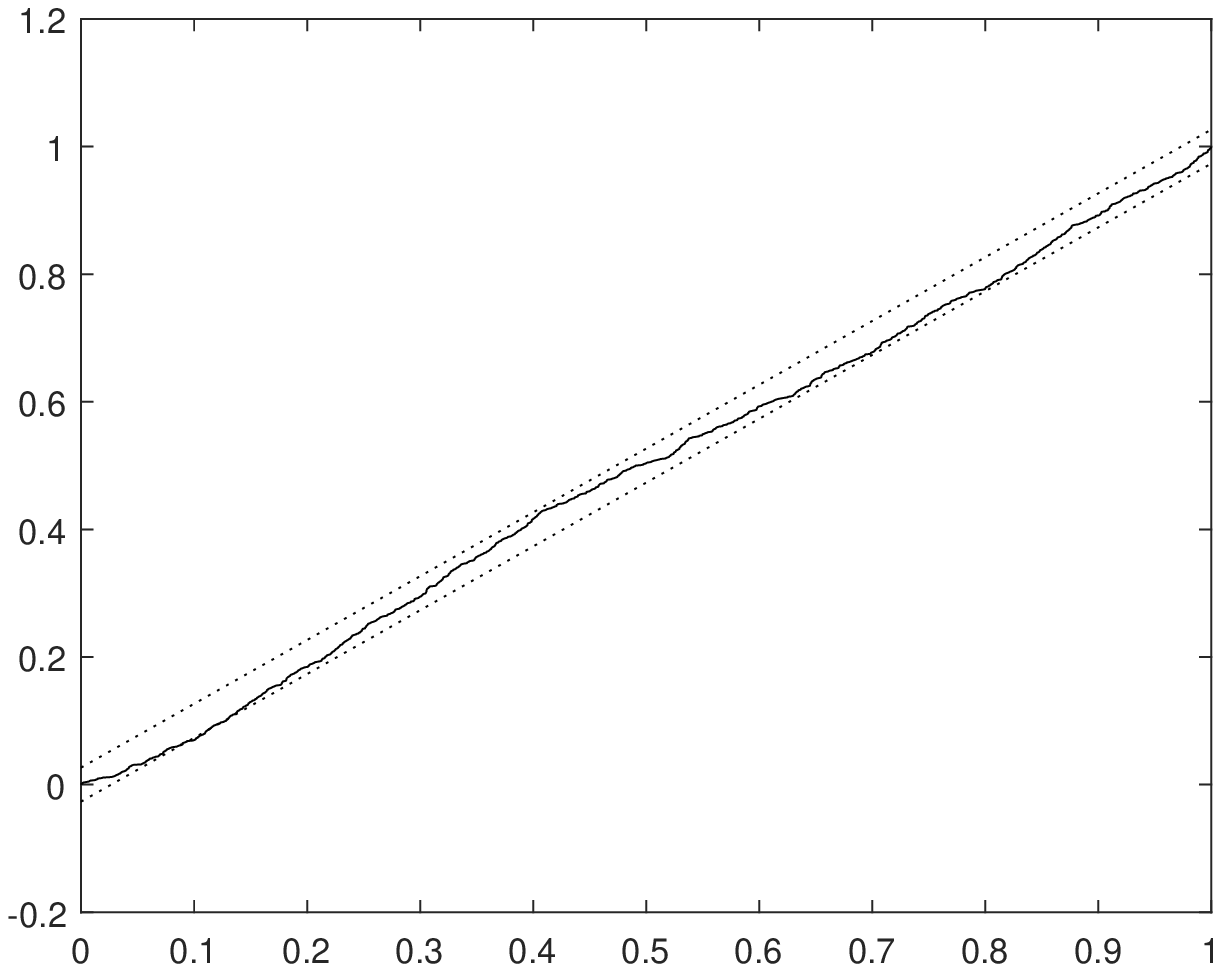}
\includegraphics[scale=.3]{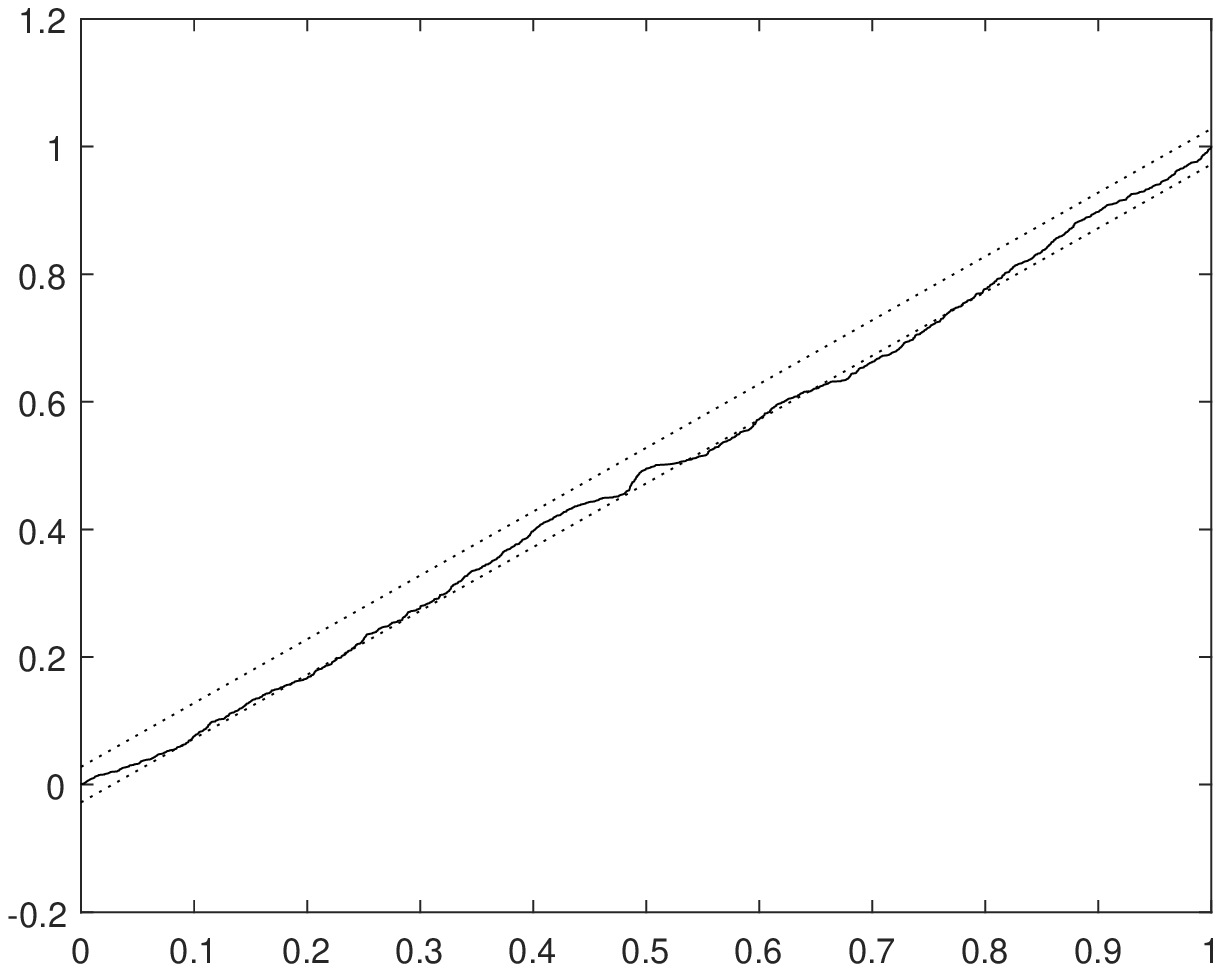}
\includegraphics[scale=.3]{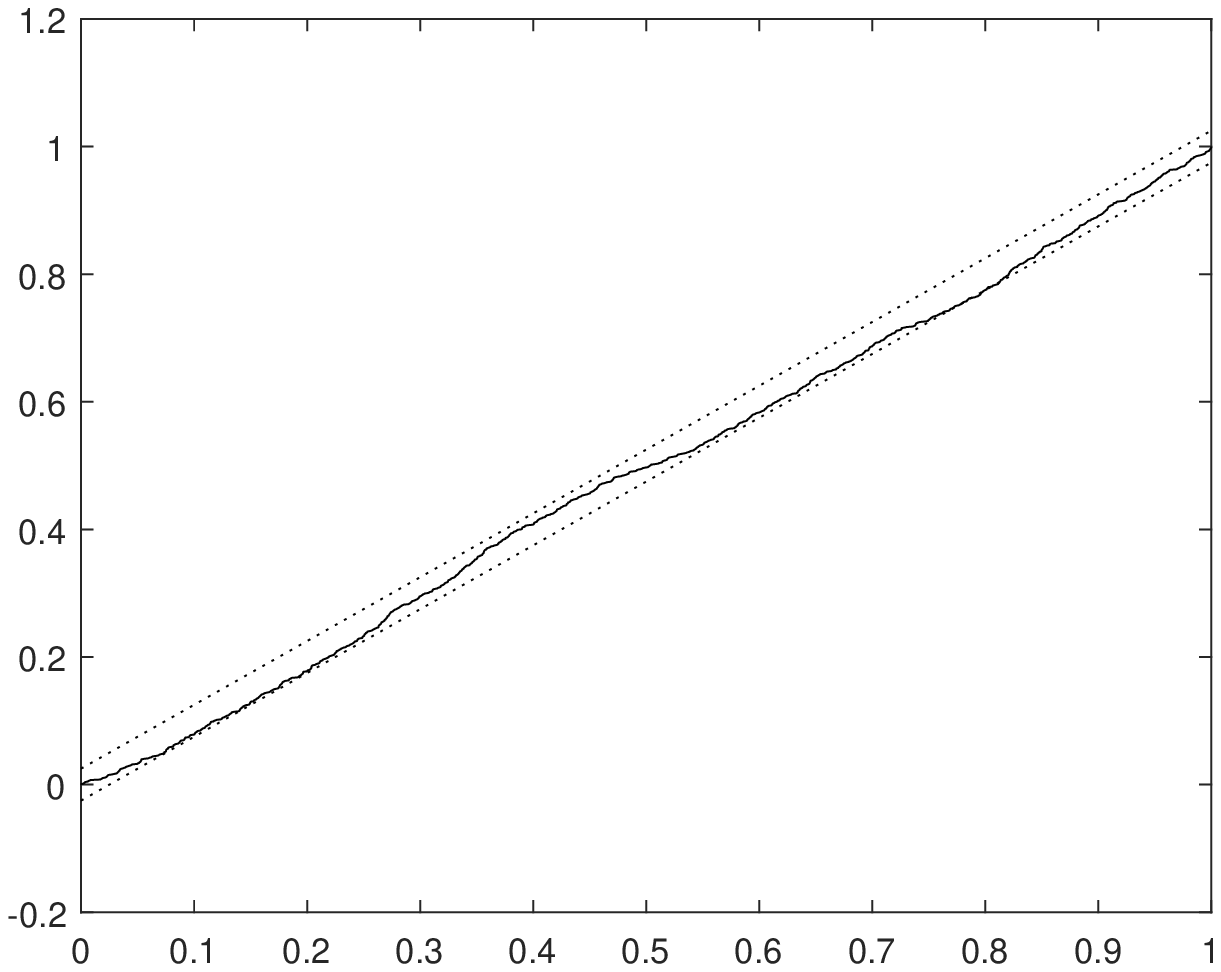}
\includegraphics[scale=.3]{DASH_GED.eps}
\includegraphics[scale=.3]{BTC_GED.eps}
\label{fig:Ged95}
\end{figure}

\newpage

\section{Discussion}

A researcher's understanding of financial asset price volatility has, for the most part, to be deduced from volatility proxies, as volatility itself is inherently unobservable. Good proxies improve parameter estimation for time volatility models. While credible parameter estimation is important, it is not an end in itself. The search for optimal proxies is beneficial to pricing financial instruments and risk management. Understanding the nature of such proxies is key for many financial applications, including asset pricing and risk management.

So far, we have presented a critique of the econometric specification which has been recently proposed in \cite{Chu17}. The authors of \cite{Chu17} fitted twelve GARCH type models and the the distribution of the innovation process were taken to be one of normal, skew normal, Student's $t$, skew Student's $t$, skew generalized error distribution, inverse Gaussian, and generalized hyperbolic distribution. Model selection criteria were then used to pick the best fit.

They found that Gaussian innovations provided the smallest values of AIC, AICc, BIC, HQC and CAIC for each cryptocurrency and each GARCH-type model. Among the twelve best fitting GARCH type models, the IGARCH (1, 1) model with normal innovations gives the smallest values of AIC, AICc, BIC, HQC and CAIC for Bitcoin, Dash, Litecoin, Maidsafecoin and Monero. The GJR-GARCH (1, 1) model with normal innovations gives the smallest values of AIC, AICc, BIC, HQC and CAIC for Dogecoin. The GARCH (1, 1) model with normal innovations gives the smallest values of AIC, AICc, BIC, HQC and CAIC for Ripple. The best fitting models were then used to provide, in their view acceptable, estimates of value at risk. The practicality of taking such an approach is open to question. There are several potential issues that are apparent.

First, \cite{Chu17} did not test whether innovations are Gaussian. To check if the innovations are Gaussian, a test of goodness-of-fit has been proposed by \cite{Bai03}, who developed a Kolmogorov-Smirnov type test of conditional distribution specifications for time series based on the comparison of an estimated conditional distribution function with the distribution function of a uniform on [0, 1]. To overcome the parameter error estimation effect, a martingale transformation is applied that delivers a nuisance-free limiting distribution for the test statistic. In order to address these concerns, we follow \cite{Remi13} and checks whether innovations are Gaussian first, before applying any model selection criteria. To this end, we deploy the Kolmogorov-Smirnov and Cramer-von-Mises test statistics. 

Second, \cite{Chu17} make loose assumptions when considering the distribution of the innovation process, namely Student's $t$ and skewed Student's $t$. Care must be taken when using the (standardized) $t$ for financial applications since its moment generating function does not exist. The nature of option pricing necessitates the use of probability distributions which provide not only a good fit to the empirical distribution of log-returns and have all their moments defined. If the innovations had a $t$ distribution under an equivalent martingale measure, the value of a call option would be infinite.

Third, although there is often a case for including higher moments (evidence of skewness in asset returns, fat tails etc), one must proceed with caution. When we include higher order moments, we should consider the combination of 3 possible cases. First is the time-dependence of higher order moments. Second is the contemporary relationship between moments (e.g. skewness and kurtosis, variance and kurtosis). Third is the time dependent relationship between moments (e.g. skewness (t) and kurtosis (t-1), variance (t) and kurtosis (t-1)). Hence, more specifications should be considered. One must be wary of uncertainties in modelling the time-dependent structure of the underlying parameters. As was forcefully argued in \cite{Jon03}, when modelling dynamic interactions among the first four moments are considered the misspecification error will likely be more substantial, potentially resulting in misguided empirical findings.

Finally, computational complexity and burden are non-negligible in this context. The distribution is determined by parameters which are estimated by MLE using an numerical optimisation algorithm. Most parametric models employ MLE technique, mainly using the numerical optimisation algorithm to deal with the potential non-linearity and asymmetry of the likelihood function. Indeed modelling the time-varying interactions among the higher order moments obtained from the underlying asymmetric distribution function (e.g. skewed Student's t-distribution) makes it much more complicated to optimise the likelihood.

So what can we learn about this? Although one can recall Cox's dictum that all models are "wrong"\footnote{ "Essentially, all models are wrong, but some are useful." in Box, George E. P.; Norman R. Draper (1987). Empirical Model-Building and Response Surfaces, p. 424, Wiley.}, model selection is an important part of any statistical analysis, and indeed is central to the pursuit of science in general. One could argue that the first step in doing applied econometrics is to establish a philosophy about models and data analysis, and then find a suitable model selection criterion. Authors in \cite{Chu17} skip this first step: they simply run a battery of AIC-type model tests on different models. 

In general, AIC finds the most predictive model. BIC finds, with probability closer to 1 as the data increases, the "correct" model if it is in the set of models considered. However, we often do not live in that sort of world. Model selection is still an art: we use our knowledge of the problem, model selection criteria, theory, and judgement to select a model. Yet, our models are often imperfect or misspecified or lack full information and so we can rarely be content with just optimizing AIC/BIC.

\section{Possible directions for future work}

The literature on Markov switching models with application to asset bubbles could make a useful contribution to the empirical debate on cryptocurrency returns. For example, one interesting  approach to test for bubbles (using Markov switching process methodology) was proposed by Hall, Psaradakis, and Sola (1999) \cite{HPS99} to capture the change from a non-bubble regime to a bubble regime. \cite{Chenug15} apply the Hall-Psaradakis-Sola test and combine it with that of \cite{PS13} to investigate the existence of bubbles in the Bitcoin market, detecting a number of short-lived bubbles over the period 2010-2014. More recently, \cite{BF17} build on this result by exploring autoregressive regime switching models for a variety of economic data series, including Bitcoin, that have previously been argued to contain bubbles, with a view to establishing whether they had a common bubble signature. With some technical caveats, they find that Bitcoin prices show bubble-like characteristics. It must be noted, however, that explosive roots need not employ regime switching methods. See, for example, the tests developed in \cite{PS13}, or the methodology proposed in \cite{BK05} and \cite{SA12} - all of which do not employ Markov Switching. It is possible that this line of research could shed some further light on the dynamics of the data generating process of cryptoassets. However, this is beyond the scope of this paper and is best left for future research.

\section{Conclusion}
This paper examined the behaviour of time series properties of cryptocurrency assets using established econometric techniques for weakly stationary financial data. Checks were performed on whether innovations are Gaussian or GED by using Kolmogorov type non-parametric tests and Khmaladze's martingale transformation. The null of gaussianity was rejected at 1\% for all GARCH(p,q) models, with $p,q \in \{1,\ldots,5 \}$, for all cryptocurrencies in sample. Although a test of normality was proposed in \cite{Bai03}, an updated test was used herewith, with a computationally advantageous Gauss-Kronrod quadrature. Parameters of GARCH models were estimated with generalized error distribution innovations using maximum likelihood. For calculating P-values, the parametric bootstrap method was used as per \cite{GK2014}. In this context, there appears to be a strong empirical argument against modelling innovations under the assumption of Gaussianity. Further, there appears to be a theoretical case for using GED innovations, rather than SGED. We demonstrated that the mgf of the Skewed GED (SGED) fails to exist under some conditions. 

These results can be used to arrive at a option pricing methodology under equivalent martingale measure - something that the methodology outlined in \cite{Chu17} does not allow one to do. Such methodology for pricing options under the GARCH assumption is described in detail in \cite{D95}, \cite{D00}, and \cite{D04}. As the cryptoasset market attracts increasing attention from regulators and investors alike, the results in this paper will be important for investment and risk management purposes.

\newpage

\appendix
\section*{Appendices}

\subsection{Moments of absolutely continuous distributions}
Let $g$ be a measurable function, i.e. $\{x\in \R; g(x)\leq y  \} \in \shb $, and 
\[  
\E\{g(X)\}= \int_\R g(x) f(x) dx, \quad \textrm{if~} \int_\R |g(x)| f(x) dx < \infty.
\]
The variance of $X$ is 
\[ 
Var(X)=\E[ \{X-\E(X)\}^2]= \int_\R \{X-\E(X)\}^2 f(x) dx.
\]
The mgf is 
\[ 
M_X(u)=\E[e^{uX}]= \int_\R e^{uX} f(x) dx, \quad u \in \R 
\]
when the right-hand side is finite. If the mgf is finite on an open interval containing zero, then the $p^{th}$ moment is the $p^{th}$ derivative w.r.t. $u$, evaluated at $u=0$,
\[ 
E[X^p]= \frac{d^p}{du^p} M_X(u) \biggr\rvert_{u=0}
\]
\subsection{Deriving moments for error distributions in GARCH models}
Let us take a GARCH(1, 1) time series model for weakly stationary financial data, specified by
\[X_t=\sigma_t Z_t \] 
where $\{X_t\}$ is the observed data, $\{Z_t\}$ is the innovation process, and $\{\sigma^2_t\}$ is the volatility process specified by
\[
\sigma_t^2=\omega+\alpha_1 X^2_{i-1} +\beta_1 \sigma^2_{i-1}. 
\] 
For each distribution for $Z_t$, we give explicit expressions for $\E[Z_t]$, $\E[Z^2_t]$, $\E[Z^3_t]$, $\E[Z^4_t]$, Value at Risk $\textrm{VaR}_p[Z_t]$, and Expected Shortfall $\textrm{ES}_p[Z_t]$.

\subsubsection{Calculating the moments: Gaussian distribution}

If $Z_t$ are independent and identical Gaussian random variables with mean $\mu$ and unit variance
then
\begin{eqnarray*}
\E [Z_t] &=& \mu \\
\E [Z^2_t] &=& \mu^2+1 \\
\E [Z^3_t] &=& \mu^3+3 \mu \\
\E [Z^4_t] &=& \mu^4+6 \mu^2 +3 \\
\textrm{VaR}_p [Z_t] &=& \mu+ \Phi^{-1}(p) \\
\textrm{ES}_p [Z_t] &=& \mu p+ \phi( \Phi^{-1}(p))
\end{eqnarray*}
where $\phi (\cdot)$ is the probability density function of a standard Gaussian random variable, and $\Phi (\cdot)$ is the cumulative distribution function of a standard Gaussian random variable.

\subsubsection{Properties of Generalized Error Distribution (GED)}
If $\{Z_t\}$ are independent and identical generalized error random variables with location parameter $\mu$ and shape parameter $a$ then
\begin{eqnarray*}
\E [Z_t] &=& \mu \\
\E [Z^2_t] &=& \mu^2+ \frac{a^{2/a-1} \Gamma (3/a)}{\Gamma(1+1/a)} \\
\E [Z^3_t] &=& \mu^3+ \frac{3\mu a^{2/a-1} \Gamma (3/a)}{\Gamma(1+1/a)} \\
E [Z^4_t] &=& \mu^4+ \frac{6\mu^2 a^{2/a-1} \Gamma (3/a)}{\Gamma(1+1/a)} + \frac{a^{4/a-1} \Gamma (5/a)}{\Gamma(1+1/a)} \\
\textrm{VaR}_p[Z_t] &=& \left\{
                \begin{array}{ll}
   \mu-a^{1/a} \Big[ Q^{-1}\Big( \frac{1}{a},2p\Big)   \Big]^{1/a} &\textrm{~if~}p\leq 1/2 \\
   \mu+a^{1/a} \Big[ Q^{-1}\Big( \frac{1}{a},2(1-p)\Big)   \Big]^{1/a} &\textrm{~if~}p> 1/2
                \end{array}
              \right. \\
\textrm{ES}_p [Z_t] &=&      \left\{
                \begin{array}{ll} 
                \mu p- \frac{a^{1/2}}{2 \Gamma(1/a)} \Gamma \Big( \frac{1}{a}\frac{\mu-\textrm{VaR})^a}{a} &\textrm{~if~ VaR} \leq \mu \\
\mu p- \frac{a^{1/2}}{2}+\frac{a^{1/2}}{2 \Gamma(1/a)} \gamma \Big( \frac{1}{a}\frac{\textrm{VaR}-\mu)^a}{a} & \textrm{~if~VaR} > \mu
                \end{array}
              \right.
\end{eqnarray*}
where \[Q(a,x)=\Big(\Gamma(a)\Big)^{-1}\int^\infty_x t^{a-1} e^{-t} dt \]
is the regularized complementary incomplete
gamma function,
\[\gamma(a,x)=\int^x_0 t^{a-1} e^{-t} dt\] is the incomplete gamma function, and 
\[\Gamma(a,x)=\int^\infty_x t^{a-1} e^{-t} dt\]  is the complementary incomplete gamma function. This distribution is abbreviated by GED.

\subsubsection{Calculating the moments: Skewed Generalized Error Distribution (SGED)}

If $\{Z_t\}$ are independent and identical generalized error random variables with location parameter $\mu$ and shape parameter $a$ then
\begin{eqnarray*}
\E [Z_t] &=& \mu -\delta + \frac{C \theta^2}{k} \Big[-(1-\lambda)^2+(1+\lambda)^2 \Big] \Gamma \Big( \frac{2}{k}\Big)   \\
\E [Z^2_t]  &=& (\mu -\delta)^2 + \frac{2C(\mu-\delta)\theta^2}{k} \Big[-(1-\lambda)^2+(1+\lambda)^2 \Big] \Gamma \Big( \frac{2}{k}\Big)  \\
&& +\frac{C\theta^3}{k}\Big[-(1-\lambda)^3+(1+\lambda)^3 \Big] \Gamma \Big( \frac{3}{k}\Big)  \\
\E [Z^3_t]  &=& (\mu -\delta)^3+\frac{3C(\mu-\delta)^2\theta^2}{k}\Big[-(1-\lambda)^2+(1+\lambda)^2 \Big] \Gamma \Big( \frac{2}{k}\Big) \\
&& + \frac{3C(\mu-\delta)\theta^3}{k}\Big[(1-\lambda)^3+(1+\lambda)^3 \Big] \Gamma \Big( \frac{3}{k}\Big) \\
&& + \frac{C\theta^4}{k}\Big[-(1-\lambda)^4+(1+\lambda)^4 \Big] \Gamma \Big( \frac{4}{k}\Big) \\
\E [Z^4_t] &=& (\mu-\delta)^4+ \frac{4C(\mu-\delta)^3\theta^2}{k} \Big[-(1-\lambda)^2+(1+\lambda)^2 \Big] \Gamma \Big( \frac{2}{k}\Big) \\
&& + \frac{6C(\mu-\delta)^2 \theta^3}{k} \Big[(1-\lambda)^3+(1+\lambda)^3 \Big] \Gamma \Big( \frac{3}{k}\Big)  \\
&& + \frac{C(\mu-\delta)\theta^4}{k} \Big[-(1-\lambda)^4+(1+\lambda)^4 \Big] \Gamma \Big( \frac{4}{k}\Big) \\
&& + \frac{C\theta^5}{k} \Big[(1-\lambda)^5+(1+\lambda)^5 \Big] \Gamma \Big( \frac{5}{k}\Big) \\
\textrm{VaR}_p[Z_t] &=& \left\{
                \begin{array}{ll}
   \mu-\delta -(1+\lambda) \theta\Big[ Q^{-1}\Big( \frac{1}{k},\frac{2p}{1+\lambda}\Big)   \Big]^{1/k}, \textrm{~if~} p\leq \frac{1+\lambda}{2} \\
   \mu-\delta +(1+\lambda) \theta\Big[ Q^{-1}\Big( \frac{1}{k},\frac{2(1-p)}{1-\lambda}\Big)   \Big]^{1/k}, \textrm{~if~}p> \frac{1+\lambda}{2}
                \end{array}
              \right. \\
\textrm{ES}_p [Z_t] &=&      \left\{
                \begin{array}{ll} 
-\frac{C(1+\lambda)^2\theta^2}{k}\Gamma \Big( \frac{2}{k}, \frac{(\mu-\textrm{VaR}-\delta)^2}{(1+\lambda)^k\theta^k} \Big) , \textrm{~if~VaR} \leq \mu-\delta \\
-\frac{C(1+\lambda)^2\theta^2}{k}\Gamma \Big( \frac{2}{k} \Big)+\frac{C(1-\lambda)^2\theta^2}{k}\gamma \Big( \frac{2}{k}, \frac{\textrm{VaR}-\mu+\delta)^2}{(1-\lambda)^k\theta^k}\Big), \textrm{~if~VaR} > \mu-\delta
                \end{array}
              \right.
\end{eqnarray*}

\end{document}